\documentclass[runningheads,a4paper]{llncs}

\usepackage{amssymb}
\usepackage{graphicx}
\usepackage[boxed]{algorithm2e}
\usepackage{tikz}
\usetikzlibrary{arrows, decorations.markings}
\usetikzlibrary{fit}					
\usetikzlibrary{backgrounds}

\def\c{\hbox{-}\cdots\hbox{-}}

\tikzstyle{every node}=[circle, draw, fill=black!50, inner sep=0pt, minimum width=20pt]

\tikzstyle{input}=[circle,
                                    thick,
                                    minimum size=2.2cm,
                                    draw=black!80,
                                    fill=white!20]

\tikzstyle{input2}=[circle,
                                    thick,
                                    minimum size=1.0cm,
                                    draw=black!80,
                                    fill=white!20]

\tikzstyle{matrx}=[rectangle,
                                    thick,
                                    minimum size=1.8cm,
                                    draw=black!80,
                                    fill=white!20]

\tikzstyle{matrx2}=[rectangle,
                                    thick,
                                    minimum size=1.5cm,
                                    draw=black!80,
                                    fill=gray!20]

\tikzstyle{vecArrow} = [thick, decoration={markings,mark=at position
   1 with {\arrow[semithick]{open triangle 60}}},
   double distance=1.8pt, shorten >= 5.5pt,
   preaction = {decorate},
   postaction = {draw,line width=1.4pt, white,shorten >= 4.5pt}]
\tikzstyle{innerWhite} = [semithick, white,line width=1.4pt, shorten >= 4.5pt]

\tikzstyle{background}=[rectangle,
                                                fill=gray!10,
                                                inner sep=0.2cm,
                                                rounded corners=5mm]

\begin{document}

\title{Learning how to rank from heavily perturbed statistics - digraph clustering approach}

\author{Krzysztof Choromanski}

\institute{Google Research, New York NY 10011, USA}

\maketitle

\begin{abstract}
Ranking is one of the most fundamental problems in machine learning with applications in many branches of computer 
science such as: information retrieval systems, recommendation
systems, machine translation and computational biology. Ranking objects based on possibly conflicting preferences is a 
central problem in voting research and social choice theory.
In this paper we present a new simple combinatorial ranking algorithm adapted to the preference-based setting.
We apply this new algorithm to the well-known scenario where the edges of the preference tournament are determined by the 
majority-voting model. It outperforms existing methods when it cannot be assumed that there exists global ranking of good enough quality
and applies combinatorial techniques that havent been used in the ranking context before.
Performed experiments show the superiority of the new algorithm over existing methods, also over
these that were designed to handle heavily perturbed statistics.
By combining our techniques with those presented in \cite{mohri}, we obtain a purely combinatorial algorithm that answers 
correctly most of the queries in the heterogeneous scenario,  where the preference tournament is only locally of good quality but 
is not necessarily pseudotransitive. As a byproduct of our methods, we obtain the
algorithm solving clustering problem for the directed planted partition
model. To the best of our knowledge, it is the first purely combinatorial algorithm tackling this problem.
 
\end{abstract}

\section{Introduction}

\subsection{Background}

The problem of ranking arises in many important applications of computer science such as information retrieval systems (e.g. the design of modern search engines), recommendation systems, computational biology and many more. There are two main approaches to the ranking problem. In the \textit{score-based setting} the input is a sample of pairwise preferences
from the dataset. The goal is to learn the so-called \textit{scoring function} $f:U \rightarrow R$ inducing a linear ordering on the set of all the objects $U$. Several algorithms were
proposed here. This setting was considered for example in \cite{freund} and \cite{rudin}. In \cite{Joachims} an SVM-based ranking algorithm for this scenario was presented.  Other algorithms include PRank given by \cite{crammer} and \cite{agarwal}. In this paper we focus on the \textit{preference-based setting} though. In this setting what is given is a preference function
$h: U \times U \rightarrow R$ taking values from the interval [0,1]. For a pair $(x,y) \in U \times U$ the closer $h(x,y)$ to 0 the more confident we are that $x$ is "better" than $y$ and vice
versa. Therefore the values of $h$ may be interpreted as probabilities. Notice that such a function induces a directed graph (\textit{digraph}) with weighted edges, where
the weights of edges are taken from the interval [0,1].
The goal is to find high-quality consistent rankings from such pairwise observations.
From now on we call the aforementioned graph a \textit{preference graph of $h$} or simply: \textit{a preference graph}. When this directed graph is a tournament (i.e. all the edges are defined),
as it will be the case in our setting, we call this graph \textit{a preference tournament}.
This approach to ranking was introduced in \cite{cohen} and led to several interesting results (\cite{balcan}, \cite{cheng}, \cite{mohri}).
Somewhat similar model was considered also in \cite{cossock}.
Notice that $h$ does not need to induce a linear ordering. In particular, the preference graph may not be a dag (i.e. it may contain directed cycles). In the tournament
setting this means that a preference tournament does not have to be transitive. This is motivated by real data.
The collection of pairwise preferences from which the preference graph is constructed may be aggregated from several noisy sources and, therefore, 
some preferences may give rise to inconsistencies or contradictions. 
For instance, the pairwise preferences taken in aggregate may not induce a consistent ranking over all the objects. 
Possibly conflicting preferences give rise to many directed cycles in the preference graph. 
As a result, the preference graph itself may be very far from being a dag. This implies that there may not exist a global good-quality scoring function. 

There are several results proposing ranking of objects in the setting where the preference tournament is not consistent but the notion of the global ranking of good quality 
makes sense (the so-called \textit{pseudotransitive setting}).
For definiteness let us assume right now that the preference graph under consideration is unweighted, i.e. all existing edges have weight $1$.
In this scenario the goal is usually to find an ordering of the vertices of the preference graph that induces as few backward edges as possible. 
Investigating all possible permutations of the set of vertices of the preference graph is usually (when the set of objects to rank is very large as it will be in  our scenario) 
untractable. The problem of finding
the permutation of vertices of a given digraph that minimizes the size of the set of backward edges, which in the literature is called a \textit{feedback arc set problem}, is NP-hard.
However there exist several approximation algorithms that output orderings with not too many more backward  edges (see for example: \cite{alon2}).
A significant breakthrough was done in  \cite{ailon} where a simple 3-approximation random algorithm for the feedback arc set problem working in $O(n\log(n))$ time was given,
where $n$ is the number of vertices of a given tournament.
The novel and counterinuitive idea was to use a quick-sort approach with pivot points chosen at random for the input graph that does not necessarily have a linear ordering of vertices.
All those results can be generalized to the weighted setting. In that case the reasonable objective function to work with is the sum of weights of backward edges.
This variation, as mentioned earlier, models the scenario where the set of different pairwise preferences expresses heterogeneous certainty level or 
heterogeneous importance. This setting is known as the weighted feedback arc set problem. Many formal results regarding  this problem  were proved by \cite{mathieu} and \cite{rudra}.
Such a problem was also considered in \cite{mohri}, where it was showed how to extend the quick-sort approach to the general weighted preference tournaments with weights taken from the interval [0,1].

\subsection{Our contribution - strongly heterogeneous setting}

Our results should be viewed as a further extension of the purely combinatorial approach from \cite{mohri} for the setting when 
optimizing the size/weight of the set of backward edges is not the right thing to do and thus the methods discussed before fail. 
As we have already noticed, the statistics that are given as an input to the ranking algorithm may be heavily perturbed. 
This makes learning the global ranking very difficult if not impossible in practice. 
All methods discussed so far may suffer from significant inconsistences and noise added to the input data.
If there does not exist a global ranking of good quality (i.e. if the assumption that a preference tournament is 
pseudotransitive is not legitimate) the need arises to find local good quality rankings. Thus every ranking algorithm needs first
to cluster the preference tournament into locally pseudotransitive chunks (i.e. chunks that can be made transitive after reversing only few directed edges) and then perform ranking algorithms separately on each chunk.
The clustering becomes a necessary preprocessing step.

We give in this paper the first purely combinatorial clustering algorithm in the directed setting that partitions preference tournaments into small number of pseudotransitive 
clusters. We combine it with the existing ranking methods to obtain new effective framework for ranking with heavily perturbed preference tournaments.
We also conduct extensive evaluation of this clustering+ranking paradigm by comparing our approach with several state-of-the-art techniques, also those
that focus on the setting with heavily perturbed statistics.

Our results can be applied in many different ways. One natural application regards the majority-voting model which is widely used to obtain the preference tournament.
In this setting users vote to determine which one from the pair of objects should get higher rank and the majority decides. Different pairs of objects attract different sets of users 
and the number of votes reflects the demand for the right evaluation of the given pair. The heterogeneity here may be implied by the fact that it does not make sense
to compare objects belonging to different categories/domains (such as favourite cars with favourite movies) or simply there is not enough data to precisely compare
objects from different categories. Those categories however do not always have to be obvious in advance and
may depend on the characteristic of the users. Thus any algorithm that aims to rank in this scenario needs also to learn the categories with good precision since only ranking within a given
category is meaningful. The algorithm should not assume that a domain is known even for a single data point. 
The exact number of groundtruth domains as well as their sizes (that may differ) are not necessarily known in advance.

After learning from the preference tournament, the ranking engine receives a stream of queries from the users and needs to correctly answer them.
Each query is taken from the same distribution that was used to construct the preference tournament and is of the form $\{u_{1},u_{2}\}$, where $u_{1},u_{2} \in \mathcal{U}$ 
are taken from the universe of all the objects. The answer indicates which object has higher rank. 
We show that our algorithm may be easily applied in this setting to answer correctly most of the queries while the other approaches fail.
The preference tournament model arising here is an example of the more general \textit{planted partition model} of the preference tournament.
This general model is a subject of our theoretical analysis. 
Planted partition model (that gained attention because of its applications in many fields of applied computer science) 
was extensively studied in the context of clustering undirected graphs (see section below) but not too many results regarding 
the directed setting are known.

\subsection{Related work - heterogeneous setting and noisy statistics}

The most straightforward way to analyze the heteregeneous setting described above is the planted partition model. 
The planted partition model is usually considered in terms of undirected graphs but there is an analogous directed formulation.
Several algorithms to reconstruct the groundtruth clustering that was used to obtain planted partition model were considered.  
Many of them use spectral partitioning techniques. Some of the most notable approaches are those of  \cite{mcsherry}, where perturbation 
theory techniques from \cite{vu} were applied as well as the results of \cite{chaudhuri}. Those results consider however mainly undirected setting 
where the domains induce dense graphs and there are not too many edges between different domains. Much less research was done in the
directed setting. In \cite{meila} the clustering with the idea of weighted cuts was considered. 
It has to be emphasized that all the papers touching the problem of clustering directed networks (see also: \cite{craven}, \cite{zhou}) have a very different goal than our 
clustering algorithm. In all these approaches a strongly connected component is considered to be a good cluster. It does not make sense in our setting, where
the entire preference tournament is with high probability strongly connected and clusters are in fact related to subtournaments that are very far from being strongly connected.
There were other papers discussing learning how to rank in the noisy setting such as \cite{feige}, where the noisy decision tree is the subject of analysis, or \cite{yue}, where noisy comparisons between pairs of strategies are performed. Both settings are substantially different from ours. In particular, none of them solves the clustering problem for directed graphs that is unavoidable in our scenario. Some of the most effective methods to rank, also with preference tournaments and for heavily perturbed statistics, are presented in  \cite{airola} and \cite{klautau}.
Those methods will be compared with our approach in the experimentals section (see: Appendix).

This work is organized as follows:

\begin{itemize}
 \item In Section 2 we formally define the heavily perturbed statistics setting as a directed planted partition model. 
       We describe the problem that needs to be solved by the ranking algorithm in this setting
       and the majority-voting model as its very special case.
 \item In Section 3 we present our ranking and clustering algorithms.
 \item In Section 4 we present all the theoretical results.
 \item In Section 5 we give final conlcusions and discuss future work.
 \item In the Appendix we give all the proofs, show experimental results,
         explain why our techniques may be also easily applied in the weighted setting, finally - comment more on  the algorithms and ineffectiveness
         of the previous methods.
\end{itemize}

\section{The model}

\subsection{Planted partition model for heavily perturbed statistics}

Assume that we are given a tournament $T$ with the set of vertices $V(T)=\mathcal{G}$, where:
$\mathcal{G} = \mathcal{G}_{1} \cup ... \cup \mathcal{G}_{k}$. We call each $\mathcal{G}_{i}$ a \textit{domain}.
We denote $n_{i}=|\mathcal{G}_{i}|$ for $i=1,...,k$.
Every set $\mathcal{G}_{i}$ contains a preferred ordering of vertices that from now on will be called
\textit{the canonical ordering of $\mathcal{G}_{i}$} and will be denoted as $\theta_{i}$. 
The directions of edges of $T$ are chosen independently according to the following procedure.
For $u_{1}, u_{2} \in \mathcal{G}_{i}$ a directed edge $(u_{2},u_{1})$ is chosen with probability $p_{i}$ ($p_{i} \ll 1$)
if $u_{1}$ appears earlier than $u_{2}$ in $\theta_{i}$ and with probability $1-p_{i}$ otherwise. 
For $u_{1} \in \mathcal{G}_{i}$, $u_{2} \in \mathcal{G}_{j}$ ($i \neq j$) a directed edge $(u_{1},u_{2})$ is chosen with probability $p_{i,j}$
and a directed edge $(u_{2},u_{1})$ is chosen with probability $p_{j,i} = 1 - p_{i,j}$ ($p_{i,j} \gg 0$).
The publically available parameters of the model are: 
\begin{itemize}
\item the upper bound $p_{u}$ on each $p_{i}$,
\item the lower bound $p_{m}$ on each $p_{i,j}$ and,
\item the upper bound $k_{u}$ on the number of domains $k$.
\end{itemize}

We call the ratio $\frac{p_{m}}{p_{u}}$ 
the \textit{heterogeneity level of the preference tournament $T$} and denote it shortly by $het(T)$.

Let us comment on this planted partition model for the preference tournament.
The sets $\mathcal{G}_{i}$ will be called by us: \textit{groundtruth domains}.
The canonical ordering models the fact that within each domain there exists a good quality ranking that
with very high probability induces only few backward edges (assumption: $p_{i} \ll 1$).
The fact that the statistics regarding objects from different domains are inconsistent (and generally of much weaker quality) 
is modeled by the fact that there exists a nontrivial lower bound
$p_{m}$ on each $p_{i,j}$.
Of course in the planted partition model we assume that  $p_{m} > p_{u}$, i.e. $het(T) > 1$. 
The larger the value of $het(T)$ is, the more heterogeneous the setting is
with the quality of statistics significantly differing for different pairs.

The objective of the ranking algorithm in this setting is to: preprocess data to get a good approximation of the groundtruth 
clustering and then to learn within each reconstructed cluster.
Our novel contribution regards the preprocessing phase.
Most known algorithms operated on the planted partition model need the exact knowledge of the parameters of the model.
In our algorithms we will just need some nontrivial bounds $p_{m}$,$p_{u}$.

We say that a set $X$ is \textit{$(1-\epsilon)$-pure} if all but at most a fraction $\epsilon$ 
of all the points from $X$ are from the same groundtruth domain. We say that a set $\mathcal{P}$ 
of the sets of vertices is \textit{$(1-\epsilon)$-pure} if every member of $\mathcal{P}$ is
$(1-\epsilon)$-pure. The goal is thus to find an $(1-\epsilon)$-pure partitioning $\mathcal{P}$
of most of the vertices of $T$ with not too many parts and for small enough $\epsilon$ 
(since one can always output as a $1$-pure partitioning a set of singletons).

\subsection{The majority-voting model}

This model is an important practical application for our algorithm and a motivation for the planted partition 
model of the preference tournament. It is one of the most popular ways to construct the preference tournament.
We formally define it now.

Let $\mathcal{U}=\{\mathcal{U}_{1},...,\mathcal{U}_{k}\}$ be the universe of all the objects partitioned into $k$ domains:
$\mathcal{U}_{1},...,\mathcal{U}_{k}$. 
Assume that there exists a global groundtruth ranking of all the objects.
The preference tournament is constructed simply by collecting statistics regarding
every unordered pair of different points from $\mathcal{U}$ (this is the training set) and choosing for each pair the preference that was given by the
majority of the users. 
Different pairs may be ranked by different users, in particular the sizes of the sets of statistics will vary from pair to pair.

Let $\mathcal{P}_{\mathcal{U}}$ be the probability distribution on the set of all unordered
pairs of different points from $\mathcal{U}$. It defines the probability that a specific pair $\{x,y\}$ will be evaluated
by the next user (in the training phase) or will be requested by the next user to be evaluated (in the test phase). 
The users choose pairs of points to evaluate/ask for evaluation independently.
Each training point consists of an unordered pair of objects for the evaluation and the evaluation itself.
Objects within a domain are compared much more frequently than between the domains.
We say that a training set $\mathcal{T}$ is $(M,m)$-unbalanced in respect to the partitioning $\{U_{1},...,U_{k}\}$
if every pair of different points from the same $\mathcal{U}_{i}$ was evaluated at least $M$ times in the training phase
and every pair of points from different: $\mathcal{U}_{i}$, $\mathcal{U}_{j}$ was evaluated at most $m$ times
in the training phase. The bigger $M$ and smaller $m$, the more heterogeneous setting we consider.
Given a pair of objects, a single user in the training phase gives a correct comparison (i.e. consistent with the groundtruth ordering)
with probability $p_{succ} > \frac{1}{2}$. 
The objective is to come up with the algorithm that gives correct answers to as many queries from the test set as possible. 

The threshold $p_{succ} > \frac{1}{2}$ is a standard assumption in all ranking models that are based on many 
independent votes. It guarantees that the sufficient number of votes will enable the algorithm to predict the right comparison with very 
high probability. In our model however not all the pairs will get the sufficient number of votes and this is where the planted partition model
of the preference tournament described in the previous section comes into action.
If we define by $\mathcal{M}$ the majority-voting model presented above and by $T_{\mathcal{M}}$ a related
preference tournament then the latter is constructed from the planted partition model introduced in the previous section.
The parameters $p_{u}$ and $p_{m}$ of $T_{\mathcal{M}}$ can be easily derived from the parameters $p_{succ}$, $M$ and $m$ of the majority-voting model
(details in Section \ref{app_2}). It turns out that we can use our clustering algorithm as a preprocessing step performed on that
preference tournament and then combine it with existing ranking methods to answer correctly most of the queries.
As we will see in the experimental section (see: Appendix), we outperform the state-of-the-art methods that can be applied in this scenario.

\section{The Algorithm}

In this section we present both: the clustering algorithm for the digraph planted partition model of tournaments and the ranking algorithm for heavily perturbed statistics.

\vspace{-0.31in}
\begin{algorithm}
\textbf{Algorithm 1 - HeteroRanking} \\
\textbf{Input:} Preference tournament $T$ with parameters:\\ $\:\:\:\:\:\:\:\:\:\:\:\:\:\:\:\:p_{u}, p_{m}, k_{u}$ and a precision parameter $\epsilon$\\
\textbf{Output: }Partitioning: $\mathcal{P}=\{\mathcal{P}_{1},...,\mathcal{P}_{t}\}$ of all but at\\$\:\:\:\:\:\:\:\:\:\:\:\:\:\:\:\:\:\:\:\:$most $\epsilon$-fraction of $V(T)$ and orderings:\\ $\:\:\:\:\:\:\:\:\:\:\:\:\:\:\:\:\:\:\:\:\sigma_{1},...,\sigma_{t}$ of $\mathcal{P}_{1},...,\mathcal{P}_{t}$\\
\Begin{
  run \textit{$DagClustering(T, p_{u}, p_{m}, k_{u}, \epsilon)$} to obtain a partitioning $\mathcal{P}$\;
  run \textit{$Purify(\mathcal{P}, p_{u}, p_{m}, \epsilon)$} to obtain $\mathcal{R}$\;
  for every $\mathcal{P}_{i} \in \mathcal{P}$ run \textit{$QuickSort(T|\mathcal{P}_{i}, \mathcal{R} \cap \mathcal{P}_{i})$} to obtain an ordering within each cluster\;
 }
\end{algorithm}
\vspace{-0.31in}

The ranking algorithm (\textit{HeteroRanking}) uses clustering subroutine (\textit{DagClustering}) and the so-called \textit{Purify} subroutine (responsible for getting rid of 
outliers from the clusters of the learned clustering)
and orders the vertices within 
each part of the obtained partitioning  $\mathcal{P}$. 
The ordering is performed by the \textit{QuickSort} subroutine from \cite{mohri}
that uses as pivot points only points from the set $\mathcal{R}$ of "non-outliers" constructed by the \textit{Purify} procedure
(for a subset $\mathcal{X} \subseteq V(T)$ we denote by $QuickSort(T,\mathcal{X})$ the algorithm from \cite{mohri} applied to the tournament $T$, but with
pivot points taken from $\mathcal{X}$ instead of $V(T)$).
When the clustering and ordering of vertices within each part of $\mathcal{P}$ is done then the mechanism of answering queries is as follows:
if the query $(x,y)$ satisfies: $x, y \in \mathcal{S}_{i}$, where $\mathcal{S}_{i} \in \mathcal{P}$, then output a point according to the computed ordering of $\mathcal{S}_{i}$.
Otherwise answer randomly.
The clustering algorithm (\textit{DagClustering}) uses the so-called \textit{gadget} structure $H$. Gadget is a small pseudo-random tournament. The only property that we want
the gadget to satisfy is to have at least one backward edge under every ordering of every subset $S \subseteq V(H)$ of size $|S| \geq \frac{|H|}{k_{u}}$. A random tournament is a gadget with high proability (details in Section \ref{app_5}). Thus gadget can be trivially constructed in advance before the main clustering algorithm starts.
There are also standard deterministic constructions of gadgets (the so-called \textit{quadratic residue tournaments}, see \cite{spencer}).  

The algorithm uses also \textit{Find} procedure, which is essentially a wrapper for the \textit{Searcher} subprocedure. It takes as an input a digraph $T_{1}$ (a subgraph of $T$)
and tries to find a special embedding of $H$ in $T_{1}$. 
It either finds this embedding (if this is the case the procedure returns the copy $H_{c}$ of $H$)
or returns two sets: $X, Y$. 
The directed density $d(X,Y)$ from $X$ to $Y$ is very close to $0$ or $1$. This, as we will see later,
implies that with very high probability most of the vertices of $X \cup Y$ came from the same groundtruth domain. In other words, we obtained a set of vertices of very good
purity. Thus the algorithm uses the local property of not having a particular pattern as a subtournament to reconstruct a significant part of the groundtruth cluster.
This set is then added to the appropriate cluster of the partial clustering that was already calculated (or potentially forms a new cluster).
The embedding we are looking for in \textit{Searcher} is a very simple one, where each vertex is being looked for in the different part
of the random partitioning of vertices into $h$  equal-length chunks.

\vspace{-0.31in}
\begin{algorithm}
 \textbf{Algorithm 2 - DagClustering} \\
 \textbf{Input:} Preference tournament $T$ with parameters:\\ $\:\:\:\:\:\:\:\:\:\:\:\:\:\:\:\:p_{u}, p_{m}, k_{u}$ and a precision parameter $\epsilon$\\
 \textbf{Output:} Partitioning: $\mathcal{P}=\{\mathcal{P}_{1},...,\mathcal{P}_{t}\}$ of all but at\\$\:\:\:\:\:\:\:\:\:\:\:\:\:\:\:\:\:\:\:\:$most $\epsilon$-fraction of $V(T)$ \\
 \Begin{
  let $H = H(k_{u})$ be a gadget\;
  let $T_{1} = T$ and $\mathcal{P} = \emptyset$\;
  \While{$|T_{1}| \geq \epsilon |T|$}{
    run \textit{$Find(H, T_{1}, \epsilon, p_{m})$}\;
    \eIf{\textit{Find} returns a copy $H_{c}$ of $H$}{
      delete all the edges of $H_{c}$ from $T_{1}$\;
      } {
       let $Z = X \cup Y$, where $X, Y$ are the sets output by \textit{Find}\;
       let $\mathcal{S}_{i} = Z \cup \mathcal{P}_{i}$ for $i=1,...$, where $\mathcal{P}=\{\mathcal{P}_{1},...\}$\;
       let $back=(\frac{het(T)}{6}+2\epsilon)|Z||\mathcal{P}_{i}|p_{u}$\;
       let $i$ be the smallest index for which
       $QuickSort(T|\mathcal{S}_{i})$ outputs an ordering with no more than $back$ backward edges in T with one endpoint in Z and the other in $\mathcal{P}_{i}$\;
       \eIf{$i$ exists}{
          replace in $\mathcal{P}$ cluster $\mathcal{P}_{i}$ by $\mathcal{S}_{i}$; 
        } {
          update: $\mathcal{P} \leftarrow \mathcal{P} \cup \{Z\}$; 
        }
        update: $V(T_{1}) \leftarrow V(T_{1}) \setminus Z$\; 
      }
  }
  output $\mathcal{P}$ \;
 }
\end{algorithm}
\vspace{-0.31in}

The \textit{Purify} subroutine gets as an input a partitioning, where each part is a good approximation of the groundtruth cluster and 
eliminates outliers from each cluster. This can be effectively done by observing that outliers contribute in a much bigger extent to the total 
number of directed triangles of a particular type in the cluster than other nodes. 
Notice that the \textit{Purify} procedure is not used by the digraph clustering algorithm. Since it does not shed any light
on our main contribution in this paper - the clustering algorithm, we will comment more on that procedure in the Appendix.

\section{Main theoretical results}

Now we state main theoretical results regarding algorithms presented in the previous section. All the proofs are given in the Appendix.

\vspace{-0.31in}
\begin{algorithm}
 \textbf{Algorithm 3 - Find} \\
 \textbf{Input:} Tournament $H$ with $V(H)=\{v_{1},...,v_{|H|}\},$\\ $\:\:\:\:\:\:\:\:\:\:\:\:\:\:\:\:$ 
                 digraph $T_{1}$, parameters: $\epsilon,p_{m}$\\
 \textbf{Output:} A copy $H_{c}$ of $H$ in $T_{1}$ or two sets\\ $\:\:\:\:\:\:\:\:\:\:\:\:\:\:\:\:\:\:\:\:X, Y \in V(T_{1})$ \\
 \Begin{
  \textbf{Initialization:\\}
  $\:\:\:\:\:$let $h=|H|$, $n_{1}=|T_{1}|$, $c=\frac{1}{4}\epsilon p_{m}$ and $S = \emptyset$\;
  $\:\:\:\:\:$partition randomly $V(T_{1})$ into $h$ sets \\$\:\:\:\:\:W_{1},...,W_{h}$, each of size $\lfloor \frac{n_{1}}{h} \rfloor$\;
  \textbf{return $Searcher(H,\{W_{1},...,W_{h}\},S)$}\;
 }
\end{algorithm}
\vspace{-0.31in}

Our first result is about the general planted partition model for the preference tournament and shows that \textit{DagClustering} algorithm
reconstructs with very good precision groundtruth domains.

\begin{theorem}
\label{clusttheorem}
Let $T$ be a preference tournament with parameters $p_{u}, p_{m}, k_{u}$ and $k$ groundtruth domains
($k$ does not have to be publicly available). Assume that $het(T) \geq 12$, each groundtruth domain is of size at least two
and has on expectation at least $\log(|T|)$ backward edges under its canonical ordering.
Let $\epsilon$ be a precision parameter satisfying: $2h\sqrt{\frac{3k_{u}h}{het(T)}} \leq \epsilon \leq \min(\frac{1}{k_{u}},\frac{1}{4}p_{m})$. 
Then for $|T|$ large enough with probability $p_{succ} = 1-o(1)$ algorithm \textit{DagClustering} with input parameters: $p_{u}, p_{m}, k_{u}$ and $\epsilon$ outputs an $(1-\epsilon)$-pure partitioning
$\mathcal{P}=\{\mathcal{P}_{1},...,\mathcal{P}_{k^{'}}\}$ of all but at most an $\epsilon$-fraction of all the vertices
of $V(T)$ for some $0 < k^{'} \leq k$.
\end{theorem}

Next theorem gives an upper bound on the generalization error of the \textit{HeteroRanking} algorithm for the introduced majority-voting model. 
 The following is true:

\begin{theorem}
\label{ranktheorem}
Assume the majority-voting model $\mathcal{M}$.
Let $\epsilon \leq p_{m}(\frac{1}{128}p_{m}-\frac{4}{het(T)})$.
Let $\mathcal{T}$ be an $(M,m)$-unbalanced training set.
Assume that the number of objects to rank is large enough
and that the related preference tournament $T_{\mathcal{M}}$ satisfies the conditions given in the statement of
Theorem~\ref{clusttheorem}.
Let $N$ be the total number of queries asked. Then for the average preference tournament with probability $p=1-o(1)$ the ranking mechanism defined by the output of the algorithm \textit{HeteroRanking} answers correctly at least: $N\frac{M}{M+m}(1-2\epsilon)^{2}(1-4p_{u})$ queries.
The average is taken under random coin tosses from the training phase. The probability $p$ is taken under random
coin tosses from the test phase.
\end{theorem}

In the statement above we can in fact get rid of averaging since the random variables under consideration are tightly concentrated around their means. Because it follows immediately from classic concentration inequalities, we leave this check
now and present it in the Appendix.
Since we have $ M \gg m$ and $p_{u} \ll 1$,  the presented ranking scheme answers correctly most of the queries. 
In comparison, existing state-of-the-art methods succeed much less frequently in this setting.
In particular, most of them are very far from achieving a recall close to $1$. In Section \ref{app_9} we will prove it and explain in more detail why
the standard approach is very ineffective under high heterogeneity assumptions.

\vspace{-0.31in}
\begin{algorithm}
  \textbf{Algorithm 4 - Searcher}$(G, \{W_{j},...,W_{h}\},S)$\\
   \textbf{Input:} Tournament $G$ with $V(G)=\{v_{j},...,v_{h}\}$,\\ 
    $\:\:\:\:\:\:\:\:\:\:\:\:\:\:\:\:$set of subsets of vertices: $W_{j},...,W_{h}$,
                         set of vertices $S$\; 
   \textbf{Output:} a pair of disjoint sets $(X,Y)$ or a copy of the gadget $H$\;
   \Begin{
    \If {$|G| = 0$} {output a tournament induced by $S$\;}
    find in $W_{j}$ a vertex $w$ with the following property for every $i={j+1,...,h}$:\\
    $\:\:\:\:\:\:\:\:\:$ $w$ is adjacent to at least $c|W_{i}|$ vertices of $W_{i}$ if $(v_{j},v_{i}) \in E(G)$, and\\
    $\:\:\:\:\:\:\:\:\:$ $w$ is adjacent from at least $c|W_{i}|$ vertices of $W_{i}$ if $(v_{i},v_{j}) \in E(G)$\;
  \eIf{$w$ is found}{
    update $S \leftarrow S \cup \{w\}$\;
    let $N_{w}(i)$ (for $i=j+1,...,h$) be:\\ $\:\:\:\:\:\:\:\:\:$a set of outneighbors of $w$ in $W_{i}$ if $(v_{j},v_{i}) \in E(G)$ and\\
     $\:\:\:\:\:\:\:\:\:$a set of inneighbors of $w$ in $W_{i}$ if $ (v_{j},v_{i}) \notin E(G)$\;
    let $G^{r}=G|\{v_{j+1},...,v_{h}\}$ and update: $W_{i} \leftarrow N_{w}(i)$ for $i=j+1,...,h$; \\ output $Searcher(G^{r},\{W_{j+1},...,W_{h}\},S)$\;
    }{
      (by the Pigeonhole Principle) there exists a set $X \subseteq W_{j}$ of order
      $|X| \geq \frac{|W_{j}|}{h-j+1}$ and an index $i^{*} \in \{j+1,...,h\}$
      with the following property:

       $\:\:\:\:\:\:\:\:\:$ either every $x \in X$ has at most $c|W_{i^{*}}|$ outneighbors in $W_{i^{*}}$ or\\
       $\:\:\:\:\:\:\:\:\:\:\:\:\:\:\:\:\:\:\:\:\:\:$ every $x \in X$ has at most $c|W_{i^{*}}|$ inneighbors in $W_{i^{*}}$\;

      let $Y = W_{i^{*}}$. \textbf{Output:} $(X,Y)$\;
    }
    }
\end{algorithm}
\vspace{-0.31in}

\section{Conclusions and future work}

We showed new algorithm performing clustering in the digraph setting. Contrary to almost all of other results on clustering digraphs,
the goal is not to partition the tournament into pseudo-strongly-connected components, but into subtournaments that can be made transitive by reversing 
only few edges. This enables us to use the algorithm as a preprocessing phase of learning how to rank from heavily perturbed preference tournaments.
To the best of our knowledge, this is the first approach of this kind that addresses at the same time
and tightly connects two important problems of modern computer science: data clustering in the directed setting and ranking. 
As a corollary, we obtain new purely combinatorial ranking algorithm and use it to effectively rank with preference tournaments constructed according to the majority-voting model.
Experimental results show the advantage of our approach over top state-of-the-art methods.
The algorithm can be viewed as a general tool for finding local nonrandom substructures in the heterogeneous network that globally looks like a random graph.

\vspace{-0.31in}
\begin{algorithm}
 \textbf{Algorithm 5 - Purify} \\
 \textbf{Input:} A partitioning $\mathcal{P}$, parameters: $p_{u}, p_{m}$ and a precision parameter $\epsilon$\\
 \textbf{Output:} Set of nonoutliers $\mathcal{R}$  \\
 \Begin{
  $\mathcal{R} \leftarrow \emptyset$\;
  let $T^{c}$ be preference tournament with parameters $p_{u}, p_{m}$ \\
  (thus of the same characteristic as $T$ but obtained independently from $T$)\;
  \For{$\mathcal{P}_{i} \in \mathcal{P}$}  
  {
     let $threshold = \frac{(1-\epsilon)^{2}}{32}|\mathcal{P}_{i}|^{2}p_{m}^{2}$\;
    \For{$v \in \mathcal{P}_{i}$} 
    {
       let $N^{+}_{v}$ be the set of outneighbors of $v$ in $T^{c}|\mathcal{P}_{i}$\;
       let $N^{-}_{v}$ be the set of inneighbors of $v$ in $T^{c}|\mathcal{P}_{i}$\;
       choose $s =\Theta(\log(n))$ samples uniformly at random from the set $N^{+}_{v} \times N^{-}_{v}$\;
       let $r$ be the number of samples $(u,w)$ such that $(u,v) \in E(T^{c})$\;
       \eIf {$\frac{r}{s}|N^{+}_{v}||N^{-}_{v}| \geq threshold$}{
         classify $v$ as an outlier\;
       }{
          $\mathcal{R} \leftarrow \mathcal{R} \cup \{v\}$\;
       }
    }\
  }\
  output $\mathcal{R}$\;
 }
\end{algorithm}
\vspace{-0.31in}

It aims to work well for very large sets of objects for which no entire preference graph is necessarily immediately known. It achieves this goal by acting locally on the preference graph, reconstructing clustering (that will be used later on to rank) part by part. Thus the authors plan to present the parall version of the algorithm in the next paper.
It would be also interesting to use similar techniques to those presented here to propose new clustering algorithm in the undirected setting.

\bibliographystyle{unsrt}
\bibliography{ranking}

\newpage
\onecolumn


\section{Appendix}
\subsection{Introduction}
\label{app_1}

In the Appendix we will comment more on the technical $\textit{Purify}$ procedure from the main body of the paper. 
We also prove the correctness of  the \textit{HeteroRanking} and \textit{DagClustering} algorithms by proving Theorem~4.1 and Theorem~4.2.
Both theorems will be proved by showing slightly more general and technical results.
Before proving Theorem~4.2, we will remind the reader the majority-voting scheme.
Furthermore, we will show why the model we are analyzing here can be used to describe weighted preference tournament setting too.
We will also show theoretical comparison between the quality of the ranking obtained by the $\textit{HeteroRanking}$ algorithm and state-of-the-art methods.
At the very end we show results of the experiments comparing our method with state-of-the-art techniques.

Let us remind that for a directed edge $(v,w)$ in a digraph $T$ we say that \textit{$v$ is adjacent to $w$} and
\textit{$w$ is adjacent from $v$}. Alternatively, we may say that $w$ is an outneighbor of $v$ and $v$ is an inneighbor of $w$.
For a set $\mathcal{A}$ we denote by $s(\mathcal{A})$ the set of all unordered pairs of different elements from $\mathcal{A}$,
namely: $s(\mathcal{A}) = \{\{x,y\}:x,y \in \mathcal{A}, x \neq y\}$.

\subsection{Ranking via majority-voting}
\label{app_2}

We will now remind the reader the majority-voting model in the context of the preference tournament. This is just one example how heterogeneous preference tournaments,
encoding ranking statistics, may be straightforwardly created by the nonuniform data. It is probably the easiest one to describe.
For other models (that we did not focus on in this paper) we can also benefit from applying the presented algorithm
for the same reasons that will soon become obvious.

Assume that the users compare certain products that come from different domains. Products from the same domain
are being compared more often than from different domains (there are many reasons for why this might be the case, as mentioned earlier, it may even not make sense to compare different domains). 
Let us assume though that there exists some groundtruth ranking of all the objects. From what we have said so far
it is clear that this ranking will play an important role only for pairs of objects within the same domain.

\begin{definition}
Let $\mathcal{U}$ be the universe of all the objects. We denote by $\mathcal{P}_{\mathcal{U}}$ the probability distribution on $s(\mathcal{U})$ from which pairs of evaluated objects (in the training phase) or pairs of objects to evaluate (in the test phase) are being selected.
\end{definition}

Set $s(\mathcal{U})$ forms an input for users' evaluations and $\mathcal{P}_{\mathcal{U}}(\{u,v\})$ is the probability
that next collected statistic will regard objects: $u$ and $v$.
In this paper we are interested in $\mathcal{P}_{\mathcal{H}}$ that is very far from being uniform.
Assume that $\mathcal{U} = \mathcal{U}_{1} \cup ... \cup \mathcal{U}_{k}$ and most of the mass of the distribution $\mathcal{P}_{\mathcal{H}}$ is concentrated on  the set:
$s(\mathcal{U}_{1}) \cup ...\cup s(\mathcal{U}_{k})$. 
Assume furthermore that when all the statistics for an unordered pair of objects $\{u,v\}$ are collected then the direction of an
edge in the preference tournament is determined by voting. Direction of an edge
is consistent with the one given by the majority of voters.
Each statistic (each vote) is given independently at random and the probability that a user made a mistake
(i.e. gives a preference not consistent with the groundtruth ranking) is $p_{mis} < \frac{1}{2}$
which may be substantial (potentially even very close to $\frac{1}{2}$) also for objects from  the same domain. A certain user is allowed to make mistakes
but since many users will be taken into account while determining a direction of each edge in the preference tournament,
the law of large numbers saves us. 

There is a caveat here that will lead us to the preference tournament model analyzed in the main body of the paper. If the number of statistics/votes is not large enough then 
taking the majority model may not be sufficient to reconstruct groundtruth ordering. Let us quantify this last statement.
If we obtain $K$ statistics for a certain unordered pair of objects $\{u,v\}$ then standard concentration
inequalities such as Chernoff's inequality, give us: $P_{mis}({u,v}) \leq e^{-\frac{\delta^{2}}{2+\delta}Kp_{succ}}$,
where: $P_{mis}$ stands for the probability of an event that an edge between $u$ and $v$ in the preference tournament will not be consistent with the groundtruth clustering, $\delta = (\frac{1}{2}-p_{mis})$ and $p_{succ} = 1 - p_{mis}$. So if $K$ is large enough then the upper bound $p_{u}$ on the probability of the mistake
will be small. However if $K$ is not too large, it may turn out that not only $p_{mis}$ but even $P_{mis}$ will be significant (this could be the case in particular when $p_{mis}$ is very close to $\frac{1}{2}$).
In this case both the probability that a direction of an edge in the preference tournament will be right and wrong
are lower-bounded by some substantial $p_{m}$. This is how the parameters $p_{u}$ and $p_{m}$ of the preference tournament come into action. They reflect the nonuniform distribution $\mathcal{P}_{\mathcal{H}}$.
We will give now a full definition of the $(M,m)$-unbalanced training set that was used in the main body of the paper.

\begin{definition}
Let $\mathcal{U}$ be the universe of all the objects and let $\{\mathcal{U}_{1},...,\mathcal{U}_{k}\}$ be the partitioning of $U$. 
Let $\mathcal{P}_{\mathcal{U}}$ be a probability distribution on 
$s(\mathcal{U})$ that describes the distribution of the elements (unordered pairs of points) used as an input for training. 
Let $\mathcal{T} \subseteq \mathcal{U} \times \mathcal{U}$ be a training set (directed pairs encode users' preferences) for which
the corresponding unordered pairs were chosen from the distribution $\mathcal{P}_{\mathcal{U}}$ and the preferences
where chosen according to the majority-voting scheme with a paramter $p_{mis}$.
We say that $\mathcal{T}$ is $(M,m)$-unbalanced with respect to the partitioning
$\{\mathcal{U}_{1},...,\mathcal{U}_{k}\}$ (or simply: $(M,m)$-unbalanced if the partitioning is clear from the context) for $M>m$ if
the following holds:
\begin{itemize}
\item $\mathcal{P}_{\mathcal{U}}(\{u,v\}) |\mathcal{T}| \geq M$ for $u,v \in \mathcal{U}_{i}$ ($i=1,...,k$), $u \neq v$ and
\item $\mathcal{P}_{\mathcal{U}}(\{u,v\}) |\mathcal{T}| \leq m$ for $u \in \mathcal{U}_{i}$, $v \in \mathcal{U}_{j}$, $i \neq j$. 
\end{itemize}
\end{definition}

In other words, we want to get on average the feedback from at least $M$ users for every pair of points from the same domain
and at most $m$ for every pair of points from different domains.

Denote the majority-voting model with the $(M,n)$-unbalanced training set described above as $\mathcal{M}$.
We denote by $T_{\mathcal{M}}$ the preference tournament model related to $\mathcal{M}$.
The parameters: $p_{u},p_{m}$ of $T_{\mathcal{M}}$ may be easily derived from the parameters $M,m, p_{mis}$ of $\mathcal{M}$ so we will not give explicit formulas here. For us it sufficies to know that the bigger $M$ and the smaller $m$, the bigger $p_{m}$ and the smaller $p_{u}$ we may take.

\subsection{\textit{Purify} procedure}
\label{app_3}

The \textit{Purify} procedure (we give it again in the Appendix for the convenience of the reader)
gets as an input a partitioning, where every part is $(1-\epsilon)$-pure with high probability.
Its goal is to get rid of the set of outliers from each part of the partitioning since the pivot points that will be used
in the \textit{QuickSort} algorithm cannot be outliers. This task can be accomplished in several different ways and is
much easier than the initial clustering problem on the directed graph since \textit{Purify} operates on the very
good approximation of the groundtruth clustering. 

\begin{algorithm}[H]
 \textbf{Algorithm 5 - Purify} \\
 \textbf{Input:} A partitioning $\mathcal{P}$, parameters: $p_{u}, p_{m}$ and a precision parameter $\epsilon$\\
 \textbf{Output:} Set of nonoutliers $\mathcal{R}$  \\
 \Begin{
  $\mathcal{R} \leftarrow \emptyset$\;
  let $T^{c}$ be preference tournament with parameters $p_{u}, p_{m}$ \\
  (thus of the same characteristic as $T$ but obtained independently from $T$)\;
  \For{$\mathcal{P}_{i} \in \mathcal{P}$}  
  {
     let $threshold = \frac{(1-\epsilon)^{2}}{32}|\mathcal{P}_{i}|^{2}p_{m}^{2}$\;
    \For{$v \in \mathcal{P}_{i}$} 
    {
       let $N^{+}_{v}$ be the set of outneighbors of $v$ in $T^{c}|\mathcal{P}_{i}$\;
       let $N^{-}_{v}$ be the set of inneighbors of $v$ in $T^{c}|\mathcal{P}_{i}$\;
       choose $s =\Theta(\log(n))$ samples uniformly at random from the set $N^{+}_{v} \times N^{-}_{v}$\;
       let $r$ be the number of samples $(u,w)$ such that $(u,v) \in E(T^{c})$\;
       \eIf {$\frac{r}{s}|N^{+}_{v}||N^{-}_{v}| \geq threshold$}{
         classify $v$ as an outlier\;
       }{
          $\mathcal{R} \leftarrow \mathcal{R} \cup \{v\}$\;
       }
    }\
  }\
  output $\mathcal{R}$\;
 }
\end{algorithm}

One possible approach focuses on the number of directed triangles touching a given point of one of the computed clusters.
If that point is the outlier we expect quadratic number of directed triangles in the cluster touching that point
with the multiplicative constant next to the quadratic factor much bigger than $\epsilon$. On the other hand, if it is not an outlier then
the expected number of directed trangles touching that point will be at most qudratic with constant next to the quadratic
factor of the order of $\epsilon$.  That observations immediately leads to the algorithm detecting outliers:
For every point of the cluster we compute the number of directed triangles touching this point and if this number
is greater that certain threshold then we classify the point as an outlier. There are two things that should be noticed here.
Counting an exact number of the directed triangles touching any given point $v$ can take $\Theta(n^{2})$ time.
Fortunately we do not need an exact number, what is really needed is a good enough approximation.
This approximation might be obtained by sampling randomly from the set of unordered pairs: ${u,w}$, where
$u \in N^{+}_{v}$, $w \in N^{-}_{w}$, and: $N^{+}_{v}$ is the set of outneighbors of $v$ and $N^{-}_{v}$
is the set of inneighbors of $v$. While sampling we count the fraction of unordered pairs ${u,v}$ such that $(u,v) \in E(T)$. If this fraction is larger than a certain threshold then we classify a point as an outlier.
To resolve the issue with a dependence between the output of the \textit{HeteroRanking} algorithm 
and the direction of edges under investigation in the \textit{Purify} subprocedure, we run \textit{Purify}
on the new preference tournament $T^{c}$ obtained independently from $T$ but for a partitioning output by the
\textit{HeteroRanking} algorithm. Tournament $T^{c}$ is obtained from the same distribution as $T$.
In the subsection where we give the proof of Theorem 4.2 we also prove correcntess of the
\textit{Purify} algorithm presented above. We will prove in particular that the number of samples $s$ needed 
is a small multiplicity of $\log(n)$.

\subsection{Tools}
\label{app_4}

We will need two standard concentration inequalities. The First one is Chernoff's inequality:

\begin{theorem}
\label{chernoff}
Let $\delta > 0$. Let $X = \sum_{i=1}^{n}X_{i}$, where $X_{i}$s are independent and each $X_{i}$ equals $1$ with probability $p_{i}$ and is zero otherwise.
Denote $\mu = EX = \sum_{i=1}^{n} p_{i}$. Then the following holds:
\begin{itemize}
\item $\mathbb{P}(X > (1+\delta)) \mu < e^{-\frac{\delta^{2}}{\delta + 2} \mu}$,
\item $\mathbb{P}(X < (1-\delta)) \mu < e^{-\frac{\delta^{2}}{\delta + 2} \mu}$.
\end{itemize}
\end{theorem}

We will also need Azuma's inequality:

\begin{theorem}
\label{azuma}
Let $\{Z_{n},n \geq 1\}$ be a martingale. Let $Z_{0}=0$. Assume that $-\alpha \leq Z_{n} - Z_{n-1} \leq \beta$ for every $n \geq 1$.
Then the following is true:
\begin{itemize}
\item $\mathbb{P}(Z_{n} \geq nc) \leq e^{-\frac{2mc^{2}}{(\alpha + \beta)^{2}}}$,
\item $\mathbb{P}(Z_{n} \leq -nc) \leq e^{-\frac{2mc^{2}}{(\alpha + \beta)^{2}}}$.
\end{itemize}
More generally, if we have: $-\alpha_{i} \leq Z_{i} - Z_{i-1} \leq \beta_{i}$ then:
\begin{itemize}
\item $\mathbb{P}(Z_{n} \leq -a) \leq e^{-\frac{2a^{2}}{\sum_{i=1}^{n}(\alpha_{i} + \beta_{i})^{2}}}$,
\item $\mathbb{P}(Z_{n} \geq a) \leq e^{-\frac{2a^{2}}{\sum_{i=1}^{n}(\alpha_{i} + \beta_{i})^{2}}}$.
\end{itemize}
\end{theorem}

\subsection{Gadget tournament $H$}
\label{app_5}

Gadget tournament is a useful well-known mathematical tool from the main body of the paper which can be easily constructed. However for the completeness we comment more on the gadget here.
In this section we briefly describel how the gadget tournament used by the \textit{HeteroRanking} algorithm should be constructed. 

We want every subset of $V(H)$ of order
at least $\frac{|V(H)|}{k_{u}}$ to have at least one backward edge under every ordering of vertices.
We call this property \textit{the gadget-property}.
Tournament $H$ can be constructed randomly as the next lemma
states:

\begin{lemma}
\label{gadgetsizelemma}
Let $H$ be a tournament satisfying: $\frac{h}{4\log(h)+1} \geq k_{u}(1-\log(1-p))$ in which the direction of every edge is chosen independently at random with probability
$\frac{1}{2}$. Then with probability at least $p$ tournament $H$ satisfies the gadget property.
\end{lemma}

\begin{proof}
Denote $\mu = \frac{1}{2} {\frac{h}{k_{u}} \choose 2} = \frac{h^{2}}{4k_{u}^{2}}(1-\frac{k_{u}}{h})$.
Denote by $p^{1}_{\delta}$ the probability that some of the $\frac{h}{k_{u}}$-element subsets of $V(H)$ induces at most
$(1-\delta) \mu$ backward edges under some ordering of vertices. Let us fix an ordering of vertices $\theta$ and lets enumerate
all the edges of the tournament. Let $X_{i}$ be an indicator random variable that is equal to $1$ if $i^{th}$ edge is backward and
is zero otherwise. If we define: $Z_{m} = \sum_{i=1}^{m} (X_{i} - 0.5)$ then we see that $\{Z_{m}:m=1,2,...\}$ is a martingale.
Besides, we have: $-0.5 \leq Z_{m} \leq 0.5$. Thus, from the Azuma's inequality, we get:
$\mathbb{P}(Z_{n} \leq -nc) \leq e^{-2nc^{2}}$, where: $n = {\frac{h}{k_{u}} \choose 2}$. Therefore we obtain: $\mathbb{P}(B \leq \mu(1-\delta)) \leq e^{-\mu \delta^{2}}$,
where $B = X_{1} + ... + X_{n}$ is a random variable counting the number of backward edges.
Now, if we sum over all $(\frac{h}{k_{u}})!$ possible orderings of vertices of the fixed subset of order $\frac{h}{k_{u}}$ and over
all possible subsets of order $\frac{h}{k_{u}}$, we obtain: 
$p^{1}_{\delta} \leq(\frac{h}{k_{u}})!{h \choose \frac{h}{k_{u}}} e^{-\mu \delta^{2}}$. Evaluating this expression, we obtain:
$p^{1}_{\delta} \leq e^{\frac{h}{k_{u}}\log(h) - \frac{\delta^{2}h^{2}}{4k_{u}^{2}}(1-\frac{k_{u}}{h})}$.
Now, if we take $\delta = 1$ and take $h$ satisfying the assumptions of the lemma, we obtain: $p^{1}_{1} \leq p$. That completes the proof.
\end{proof}

Note that in particular we have proved that for every $h$ satisfying: $\frac{h}{4\log(h) + 1} > k_{u}$ there exists an $h$-vertex tournament
satisfying gadget-property. 
In practice we even do not need to make a random construction to obtain $H$. There are plenty deterministic constructions of tournaments
with pseudo-random properties, in particular with a gadget-property. For example one can use the family of quadratic residue tournaments.
The proof given above is useful though to get a simple upper bound on the order of tournaments that may serve as gadgets.

\subsection{Proof of Theorem~4.1}
\label{app_6}

We give here detailed proof of the correctness of Theorem~4.1.
We start with the lemma that plays an important role in the procedure \textit{Find} used by \textit{DagClustering}.

\begin{lemma}
 \label{forblemma}
 Let $c>0$, let $H$ be a tournament and let $T$ be an $H$-free digraph. 
 Then $V(T)$ contains two disjoint subsets $A,B$ satisfying:
 \begin{itemize}
  \item $|A| \geq c^{h-1}\lfloor \frac{n}{h^{2}} \rfloor$, $|B| \geq c^{h-1}\lfloor \frac{n}{h} \rfloor$ and
  \item either every vertex from $A$ is adjacent to at most $c|B|$ vertices from $B$ or every vertex from $A$
        is adjacent from at most $c|B|$ vertices in $B$.
 \end{itemize}
\end{lemma}

That lemma was stated and proved in a little bit different setting (undirected graphs) in \cite{erdos0}.
Since the proof of the directed setting is very similar, we refer the reader to \cite{erdos0} for details
(Lemma 1.5, p.40).  The procedure \textit{Find} mimics the proof of the lemma above.
In particular, whenever \textit{Find} outputs two sets: $X,Y$ we have:
 \begin{itemize}
  \item $|X| \geq c^{h-1}\lfloor \frac{\epsilon n}{h^{2}} \rfloor$, $|B| \geq c^{h-1}\lfloor \frac{\epsilon n}{h} \rfloor$ and
  \item either every vertex from $X$ is adjacent to at most $c|Y|$ vertices of $Y$ or every vertex from $X$
 is adjacent from at most $c|Y|$ vertices of $Y$.
 \end{itemize}

An important conclusion from the lemma is that the absence of the tournament $H$ in $T$
implies the property that random tournaments satisfy with very small probability, namely the existence
of two substantial (linear) sets: $A,B$ with directed density between them close to one or zero.
The intuition is now that with high probability most of the vertices from these two sets came in fact from the same 
domain (note that tournaments induced by domains are very nonrandom). Thus, by getting $A$ and $B$, 
we can with high probability extract very "pure chunk". 
This chunk can be then added to the part of the related domain that has been already extracted.
We will make all these observations much more precise a little bit later.

In the \textit{DagClustering} algorithm we delete from a digraph all edges of the copy of $H$ it the copy was found.
The explanation is as follows: if the copy was found then one of its edges must be a backward edge within some domain under its 
canonical ordering (this easy observation is a consequence of the definition of the gadget tournament, we will see why 
later). We call edges like that \textit{bad edges}. We dont know exactly which edges of the copy are bad and that is why we delete all of them. 
By doing it systematically, we eventually get rid of all bad edges. Doing it we also get rid of
edges that are not necessarily bad. Fortunately, with high probability the number of bad edges
is not very large thus while clearing up the entire digraph from bad edges we get rid of not too 
many other edges. Thus the detection of the copy of $H$ and deletion of its edges from the digraph is a convenient way to 
detect bad edges without doing much harm to the overal structure of the digraph. 
The following is a useful property of gadgets:

\begin{lemma}
\label{gadgetproplemma}
 Let $T_{1}$ be a digraph from the \textit{DagClustering} algorithm. Then if $T_{1}$ contains a copy of $H$, 
 one of the edges of $H$ is a bad edge.
\end{lemma}

\begin{proof}
Assume by contradiction that the found copy of $H$ does not contain a bad edge. 
By the Pigeonhole principle, at least $\frac{h}{k_{u}}$ vertices of
the found copy of $H$ were taken from the same domain. Call this set $\mathcal{X}$.
Take a canonical ordering of the vertices from $\mathcal{X}$. From the gadget property we know that this ordering induces at least one
backward edges. That contradicts our previous assumption.
\end{proof}

Now we need to quantify the statement that a preference tournament $T$ that is an input of the algorithm
with high probability does not have too many bad edges.
Denote by $n_{1},...,n_{k}$ the sizes of the domains.

\begin{lemma}
 \label{fewbackedges}
 Let $g(n)$ be some positive function. Let $\delta = \max(2, \frac{4\log(g(n))}{\sum_{i=1}^{k}n_{i}^{2}p_{i}(1-\frac{1}{n_{i}})})$.
 Denote $M = (1+ \delta)\sum_{i=1}^{k}\frac{n_{i}^{2}p_{i}}{2}(1-\frac{1}{n_{i}})$. Then the probability $p^{2}$ that a preference 
 tournament $T$ contains more than $M$ bad edges is at most $\frac{1}{g(n)}$.
\end{lemma}

\begin{proof}
Let $B$ be a random variable that counts the number of bad edges in $T$. Note that $B$ is a sum of $\sum_{i=1}^{k} {n_{i} \choose 2}$ random
variables $X_{i}$, where each $X_{i}$ corresponds to a certain pair of vertices within a particular domain.
Every $X_{i}$ is one with probability $p_{j}$, where: $j$ is the number of the domain that the vertices corresponding to $X_{i}$ were taken from,
and zero otherwise.
Thus we have: $\mu = EB = \sum_{i=1}^{k} {n_{i} \choose 2}p_{i} = \sum_{i=1}^{k} \frac{n_{i}^{2}p_{i}}{2}(1-\frac{1}{n_{i}})$.
Now, for any $\delta>0$ the probability $p^{2}_{\delta}$ that the number of backward edges is more than $(1+\delta)\mu$, is 
(by ~\ref{chernoff}) at most $e^{-\frac{\delta^{2}}{2+\delta}\sum_{i=1}^{k}\frac{n_{i}^{2}p_{i}}{2}(1-\frac{1}{n_{i}})}$.
The expression on the LHS of the last inequality is at most $\frac{1}{g(n)}$ if:
$\delta \frac{1}{1+\frac{2}{\delta}} \sum_{i=1}^{k}\frac{n_{i}^{2}p_{i}}{2}(1-\frac{1}{n_{i}}) \geq \log(g(n))$.
One can easily notice that this inequality is satisfied for our choice of the value of $\delta$ from the statement of the lemma.
That completes the proof.
\end{proof}

Let us remind the definition of \textit{$(1-\epsilon)$-purity}. We say that a set of vertices $X$ is $(1-\epsilon)$-pure if all but at most 
an $\epsilon$-fraction of all the vertices of $X$ are from the same domain.
For two disjoint sets: $X,Y$ we denote by $E(X,Y)$ the number of directed edges going from $X$ to $Y$. Intuitively speaking,  
we expect to have substantial numbers of directed edges going from both: $X$ to $Y$ and $Y$ to $X$ if $X$ and $Y$ contain substantial chunks 
from different domains. Below we make this statement precise and give it in the form that will be very useful later in the 
proof (parameters $p_{u}, p_{m}$ used in the statement of the next lemma were already defined in the section describing preference tournament model): 

\begin{lemma}
 \label{notpurelemma}
Let $\delta_{M}, M$ be as in Lemma~\ref{fewbackedges}.
Let $T$ be a preference tournament with $|T|=n$ and each domain of size at least two.
Assume that $k_{u} \geq 2$.
 Let $\epsilon$ satisfy: $\frac{1}{k_{u}} \geq \epsilon \geq \sqrt{\frac{2(1+\delta_{M})k_{u}h^{5}p_{u}}{p_{m}}}$. Let $h>0$ and $q(n)$ be a positive function. Assume that $0 < c  \leq \frac{1}{4}\epsilon p_{m}$ and 
$n \geq \max(\frac{2h^{2}}{\epsilon}, \frac{32h^{3}\log(2)}{\epsilon^{4}c^{2h-2}p_{m}^{2}},\frac{4h}{\epsilon^{2}c^{h-1}p_{m}}\sqrt{h\log(q(n))})$.
 Denote by $T^{d}$ a tournament obtained from the preference tournament $T$ by deleting some $\frac{h^{2}}{2}M$ edges
 (notice that we do not assume anything about the mechanism according to which those edges were deleted, in particular
  the set of deleted edges might be highly correlated with the overal structure of $T$).
 Let $\mathcal{E}$ be the following event:
 \begin{itemize}
  \item  there exist two sets: $A$ and $B$ in $T^{d}$ such that: $|A| \geq c^{h-1} \lfloor \frac{\epsilon n}{h^{2}} \rfloor$,
           $|B| \geq c^{h-1} \lfloor \frac{\epsilon n}{h} \rfloor$, $A \cup B$ is not $(1-\epsilon)$-pure and either
           every vertex of $A$ is adjacent to at most $c|B|$ vertices in $B$ or every vertex in $A$ is adjacent from at most 
           $c|B|$ vertices in $B$.
 \end{itemize}

 Then the probability $p^{3}$ that $\mathcal{E}$ holds is at most $\frac{1}{q(n)}$.
\end{lemma}

\begin{proof}
Take two sets: $A$ and $B$ satisfying the property in the statement of the lemma.
From the Pigeonhole principle we know that $A$ contains a subset $A^{'}$ of 
size $|A^{'}| \geq \frac{1}{k_{u}} |A| \geq \epsilon |A|$ that belongs entirely to one of the groundtruth domains.
Denote this domain by $\mathcal{D}$. 
If $B$ does not contain at least $(1-\epsilon)|B|$ vertices from $\mathcal{D}$ then it contains at least $\epsilon |B|$ vertices from
other domains. On the other hand, if $B$ contains at least $(1-\epsilon)|B|$ vertices from $\mathcal{D}$ then $A$ contains
at least $\epsilon|A|$ vertices from other domains (since $A \cup B$ is not $(1-\epsilon)$-pure).
In both scenarios we conclude that there exist two sets: $X_{1} \subseteq A$ and $X_{2} \subseteq B$ such that no groundtruth
domain intersects both of them.
From the property of $A$ and $B$ we know that either the number of directed edges going from $X_{1}$ to $X_{2}$ in $T^{d}$
is at most $|X_{1}| \cdot c|B|$ or the number of directed edges going from $X_{2}$ to $X_{1}$ in $T^{d}$ is at most
$|X_{1}| \cdot c|B|$. Thus the number of those edges is at most $\frac{c}{\epsilon}|X_{1}||X_{2}|$.
We can conclude that an event $\mathcal{E}$ is contained in the following event $\mathcal{F}$:
there exists a pair of sets: $X_{1},X_{2}$, such that: $|X_{1}| \geq \epsilon c^{h-1} \lfloor \frac{\epsilon n}{h^{2}} \rfloor$,
$|X_{2}| \geq \epsilon c^{h-1} \lfloor \frac{\epsilon n}{h} \rfloor$ and either the number of directed edges in $T$ going from
$X_{1}$ to $X_{2}$ is at most $\frac{c}{\epsilon}|X_{1}||X_{2}| + M$ or the number of directed edges in $T$ going from
$X_{2}$ to $X_{1}$ is at most $\frac{c}{\epsilon}|X_{1}||X_{2}| + M$. 
We have: $p^{3} \leq \mathbb{P}(\mathcal{F})$.
Let us calculate now the probability of $\mathcal{F}$. 
Lets first fix $X_{1}$ and $X_{2}$ and the direction where most of the directed edges between $X_{1}$ and $X_{2}$ go.
This can be done in at most $2^{n} \cdot 2^{n}$ different ways. 
Assume, without loss of generality that the "preferable direction" is from $X_{1}$ to $X_{2}$.
For a pair $(x_{1},x_{2})$ such that: $x_{1} \in X_{1}$ and $x_{2} \in X_{2}$ denote by $Y_{(x_{1},x_{2})}$ a random
variable that is zero if there exists a directed edge from $x_{1}$ to $x_{2}$ in $T$ and is one otherwise.
Denote: $Y = \sum_{(x_{1},x_{2}) \in X_{1} \times X_{2}} Y_{(x_{1},x_{2})}$.
We know that $Y_{(x_{1},x_{2})}$ is one with probability at least $p_{m}$. Thus $EY \geq |X_{1}||X_{2}|p_{m}$.
On the other hand, from the properties of 
$X_{1}$ and $X_{2}$ we know that: $Y \leq \frac{c}{\epsilon}|X_{1}||X_{2}| + M$. 
Therefore an inequality: $Y \leq \frac{c}{\epsilon}|X_{1}||X_{2}| + M$ implies: 
$Y - EY \leq -|X_{1}||X_{2}|(p_{m} - \frac{c}{\epsilon} - \frac{M}{|X_{1}||X_{2}|})$.
Now, knowing the lower bounds on $|X_{1}|$ and $X_{2}$,using a general Azuma's inequality
(see: ~\ref{azuma}) for specific $X_{1}$ and $X_{2}$ and a union bound over all pairs $(X_{1},X_{2})$, we obtain:
$\mathbb{P}(\mathcal{F}) \leq 2^{n} \cdot 2^{n} 
\exp(-2\epsilon^{2} \lfloor \frac{\epsilon n}{h} \rfloor \lfloor \frac{\epsilon n}{h^{2}} \rfloor c^{2h-2}(p_{m}-\frac{c}{\epsilon}-
\frac{M}{ \lfloor \frac{\epsilon n}{h} \rfloor \lfloor \frac{\epsilon n}{h^{2}} \rfloor})^{2})$.
Thus we have: 
$\mathbb{P}(\mathcal{F}) \leq 2^{n} \cdot 2^{n} 
\exp(-2\epsilon^{2} \frac{(\epsilon n)^{2}}{h^{3}}(1-\frac{h}{\epsilon n})(1-\frac{h^{2}}{\epsilon n}) c^{2h-2}(p_{m}-\frac{c}{\epsilon}-\frac{M}{\frac{(\epsilon n)^{2}}{h^{3}}(1-\frac{h}{\epsilon n})(1-\frac{h^{2}}{\epsilon n})})^{2})$.
One can check that under our choice of parameters from the assumptions of the lemma we have: 
$\mathbb{P}(\mathcal{F}) \leq \frac{1}{q(n)}$. Since $\mathcal{E} \subseteq \mathcal{F}$, we also have
$\mathbb{P}(\mathcal{E}) \leq \frac{1}{q(n)}$ and that completes the proof.
\end{proof}

We need one more observation before proving Theorem~4.1.
In the clustering algorithm when we extract a set of vertices $Z$ we need to decide to which partial
cluster this set should be added (it could be also the case that $Z$ will form a new cluster). If we know that 
all the sets under consideration are pure enough then we can use this fact to make a right choice.
When we consider partial cluster $P_{i}$ we can find the ordering of vertices of $Z \cup P_{i}$ that somehow
approximates an optimal ordering with the minimum number of backward edges (this can be done for example
with the use of the \textit{QuickSort} algorithm). If the number of backward edges under this ordering is big enough then
with high probability we can conclude that $P_{i}$ and $Z$ were taken from different domains and so
$Z$ should not be added to $P_{i}$. Otherwise, with very high probability they come from the same domain
and that is why we should merge $Z$ with $P_{i}$. We make this intuitive statement more formal below: 

\begin{lemma}
\label{whenmergelemma}
Let $T$ be a preference tournament with $n=|T|$ with $\frac{p_{m}}{p_{u}} \geq 12$.
Let $0 < \epsilon \leq \frac{1}{4}p_{m}$ and $\delta_{1} = \frac{p_{m}}{6p_{u}} - 1$.
Let $w(n)$ be a positive function such that $w(n) \leq 2^{n-1}$.
Assume furthermore that $n \geq \max(\frac{2h}{c^{h-1} \epsilon},\frac{288 \log(2) h^{2}}{c^{2h-2}\epsilon^{2}p_{m}})$ and $\frac{n}{\log(n)} \geq \frac{864 h^{2}}{c^{2h-2}\epsilon^{2}p_{m}}$.
Let $p^{4}$ be the probability of the following event $\mathcal{G}$:
\begin{itemize}
 \item there exist two sets $X$, $Y$ in $V(T)$ that are both $(1-\epsilon)$-pure, $|X| \geq c^{h-1} \lfloor \frac{\epsilon n}{h} \rfloor$, $|Y| \geq c^{h-1} \lfloor \frac{\epsilon n}{h} \rfloor$, such that most of the vertices of $X$ are taken from the same domain that most of the vertices of $Y$ and there exists an ordering of $X \cup Y$ with more than 
$(1 + \delta_{1} + 2\epsilon)|X||Y|p_{u}$ backward edges in $T$ with one endpoint in $X$ and the other in $Y$ or
\item there exist two sets $X$, $Y$ in $V(T)$ that are both $(1-\epsilon)$-pure, $|X| \geq c^{h-1} \lfloor \frac{\epsilon n}{h} \rfloor$, $|Y| \geq c^{h-1} \lfloor \frac{\epsilon n}{h} \rfloor$, such that most of the vertices of $X$ are taken from a different domain that most of the vertices of $Y$ and there exists an ordering of $X \cup Y$ with at most 
$3(1 + \delta_{1} + 2\epsilon)|X||Y|p_{u}$ backward edges in $T$ with one endpoint in $X$ and the other in $Y$
\end{itemize}
Then $p^{4} \leq \frac{1}{w(n)}$.
\end{lemma}

\begin{proof}
Let us first consider two sets $X$ and $Y$ such that most of the vertices of $X$ came from the same domain as
most of the vertices of $Y$. Denote this domain by $\mathcal{D}$. Order the vertices of $(X \cup Y) \cap \mathcal{D}$
according to the canonical ordering of $\mathcal{D}$ and add the remaining vertices of $X \cup Y$ to that ordered sequence in the arbitrary way.
The number $B_{1}$ of backward edges in $T$ induced by that ordering with one endpoint in $X$, one in $Y$ and involving points not from $\mathcal{D}$ is (from $(1-\epsilon)$-purity) at most $2\epsilon|X||Y|$. 
Denote by $B_{2}$ the number of backward edges in $T$ induced by that ordering with one endpoint in $X$, one in $Y$ and involving only points from $\mathcal{D}$. By the similar analysis as in the proofs of the previous lemmas, we conclude (using ~\ref{chernoff}) that $\mathbb{P}(B_{2} > (1+\delta_{1}) \mu) < e^{-\frac{\delta_{1}^{2}}{2+\delta_{1}}\mu}$, where $\mu = |X||Y|p_{u}$. Thus the probability that:
$B_{1} + B_{2} \geq (1+\delta_{1})\mu + \frac{2\epsilon}{p_{u}}\mu$ is at most 
$e^{-\frac{\delta_{1}^{2}}{2+\delta_{1}}\mu}$. If we now sum over all possible subsets $X,Y$
with $|X| \geq c^{h-1} \lfloor \frac{\epsilon n}{h} \rfloor$, $|Y| \geq c^{h-1} \lfloor \frac{\epsilon n}{h} \rfloor$
then we get the following upper bound: $p^{4}_{a} \leq 2^{2n} \exp(-\frac{\delta_{1}^{2}}{2+\delta_{1}}
\frac{2c^{h-1} \lfloor \frac{\epsilon n}{h} \rfloor(2c^{h-1}\lfloor \frac{\epsilon n}{h} \rfloor - 1)}{2}p_{u})$.
Now let us assume that most of the vertices of $X$ are from different domain than most of the vertices of $Y$
(second scenario in the statement of the lemma). Fix some ordering of vertices and sets $X$ and $Y$.
If we denote by $B$ the number of backward edges with one endpoint of $X$ and one in $Y$ under this given ordering,
then, using similar analysis as before, we conclude that the probability that $B \leq 3(1 + \delta_{1} + 2\epsilon)|X||Y|p_{u}$ is at most
$e^{-\frac{\delta_{2}^{2}}{2 + \delta_{2}}\mu^{'}}$, where: $\delta_{2} = \frac{1}{2} - \frac{\epsilon}{p_{m}}$
and $\mu^{'}=|X||Y|p_{m}$. If we now sum over all possible orderings of vertices and all possible choices of $X$ and
$Y$ then we get the following upper bound: $p^{4}_{b} \leq 2^{2n} n! \exp(-\frac{\delta_{1}^{2}}{2+\delta_{1}}
\frac{2c^{h-1} \lfloor \frac{\epsilon n}{h} \rfloor(2c^{h-1}\lfloor \frac{\epsilon n}{h} \rfloor - 1)}{2}p_{m})$.
We have: $p^{4} \leq p^{4}_{a} + p^{4}_{b}$.
Under our assumptions on the values of parameters used in the statement of the lemma one can check that:
$p^{4}_{a} \leq \frac{1}{2w(n)}$ and $p^{4}_{b} \leq \frac{1}{2w(n)}$. (This time we will not show
the calculations in more detail since they do not involve anything more than a tedious algebra.)
That completes the proof.
\end{proof}

We are ready to prove Theorem~4.1.
We will in fact prove more general yet also much more technical result from which
Theorem~4.1 follows.

\begin{theorem}
\label{clusttechtheorem}
Let $q(n), g(n)$ be positive functions.
Assume that $T$ is a preference tournament of $n$ vertices and with parameters: $p_{u}, p_{m}, k_{u}$ and such that $het(T) \geq 12$. Let us assume that every domain of $T$ contains at least two vertices and $T$ consists of $k$ domains.
 Let $H$ be a gadget tournament used by the algorithm.
Denote $h=|H|$.
Let $\epsilon>0$ be a precision parameter. Assume that $2h\sqrt{\frac{(1+\delta_{M})k_{u}h}{het(T)}} \leq \epsilon
\leq \min(\frac{1}{k_{u}},\frac{1}{4}p_{m})$, where: $\delta_{M}=\max(2,\frac{4\log(g(n))}
{\sum_{i=1}^{k}n_{i}^{2}p_{i}(1-\frac{1}{n_{i}})})$. 
Let us assume that $\frac{n}{\log(n)}  \geq  \frac{288 h^{2}}{c^{2h-2}\epsilon^{2}p_{m}}$ and
$\frac{n}{\sqrt{q(n)}} \geq \frac{80h^{3} \log(2)}{\epsilon^{4}c^{2h-2}p_{m}^{2}}$, where:
$c=\frac{1}{4}\epsilon p_{m}$.
Then \textit{DagClustering} algorithm outputs $(1-\epsilon)$-pure partitioning of all but at most an $\epsilon$-fraction
of all the vertices of $V(T)$ with probability $P_{succ} \geq (1 - (\frac{1}{q(n)} + \frac{1}{g(n)} + \frac{1}{2^{n-1}}))(1-O(p_{f}))$, where: $p_{f}$ is a probability that the method proposed in \cite{mohri}
does not output the $3$-approximation of the feedback arc set problem.
\end{theorem}

\begin{proof}
Note that obviously during the entire execution of the algorithm every time we perform an operation on the
tournament $T_{1}$ we have: $|T_{1}| \geq \epsilon n$.

Let $M$ be as in Lemma~\ref{fewbackedges}. Let $\mathcal{A}$ be the following event:
tournament $T$ has no more than $M$ bad edges. Let $\mathcal{B}$ be the following event:
there do not exist two sets: $A$ and $B$ in $T_{1}$ during the entire execution of the algorithm such that: 
\begin{itemize}
\item $|A| \geq c^{h-1} \lfloor \frac{\epsilon n}{h^{2}} \rfloor$,
\item $|B| \geq c^{h-1} \lfloor \frac{\epsilon n}{h} \rfloor$, 
\item $A \cup B$ is not $(1-\epsilon)$-pure and 
\item either every vertex of $A$ is adjacent to at most $c|B|$ vertices in $B$ or every vertex in $A$ is adjacent from at    
        most $c|B|$ vertices in $B$.
\end{itemize}
Let $\mathcal{C}$ be the following event:
\begin{itemize}
 \item there do not exist two sets $X$, $Y$ in $V(T)$ that are both $(1-\epsilon)$-pure, $|X| \geq c^{h-1} \lfloor \frac{\epsilon n}{h} \rfloor$, $|Y| \geq c^{h-1} \lfloor \frac{\epsilon n}{h} \rfloor$, such that most of the vertices of $X$ are taken from the same domain that most of the vertices of $Y$ and there exists an ordering of $X \cup Y$ with more than 
$(1 + \delta_{1} + 2\epsilon)|X||Y|p_{u}$ backward edges in $T$ with one endpoint in $X$ and the other in $Y$ and
\item there do not exist two sets $X$, $Y$ in $V(T)$ that are both $(1-\epsilon)$-pure, $|X| \geq c^{h-1} \lfloor \frac{\epsilon n}{h} \rfloor$, $|Y| \geq c^{h-1} \lfloor \frac{\epsilon n}{h} \rfloor$, such that most of the vertices of $X$ are taken from a different domain that most of the vertices of $Y$ and there exists an ordering of $X \cup Y$ with at most 
$3(1 + \delta_{1} + 2\epsilon)|X||Y|p_{u}$ backward edges in $T$ with one endpoint in $X$ and the other in $Y$
\end{itemize}

Notice that under our choice of parameters, using lemmas:~\ref{fewbackedges}, ~\ref{notpurelemma} and 
~\ref{whenmergelemma}, we can conclude that $$\mathbb{P}(\mathcal{A}^{c} \cup \mathcal{B}^{c} \cup \mathcal{C}^{c}) 
\leq \frac{1}{q(n)} + \frac{1}{g(n)} + \frac{1}{2^{n-1}},$$ where $X^{c}$ stands for a complement of an event $X$. 
Let $\mathcal{J}$ be an event that  all of: $\mathcal{A}$, $\mathcal{B}$, $\mathcal{C}$ hold.
Then we have: $\mathbb{P}(\mathcal{J}) \geq 1 - (\frac{1}{q(n)} + \frac{1}{g(n)} + \frac{1}{2^{n-1}})$.
Now assume that $J$ holds. By Lemma~\ref{gadgetproplemma} we know that every time the subprocedure 
\textit{Find} detects a copy of $H$, one of its edges is a bad edge.  Since the total number of bad edges in $T$ is $M$
and every time a bad edge is detected a set of ${h \choose 2}$ edges of $T$ (containing this edge) is being removed,
we conclude that the algorithm removes at most ${h \choose 2} M$ edges of $T$.
We can also conclude that \textit{Find} returns a copy of $H$ at most $M$ times. 
Let us assume now that \textit{Find} returns two sets: $X, Y$. By Lemma~\ref{forblemma} and the fact that 
$\mathcal{J} \subseteq \mathcal{B}$ we conclude that $X \cup Y$ is $(1-\epsilon)$-pure.
If we now assume inductively that all $\mathcal{P}_{i}$s from the algorithm are $(1-\epsilon)$-pure,
and the procedure proposed in \cite{mohri} gives a $3$-approximation of the  feedback arc set problem,
then using the the inclusion:  $\mathcal{J} \subseteq \mathcal{C}$, we conclude that the partitioning
is $(1-\epsilon)$-pure during the entire execution of the algorithm.  The algorithm obviously terminates
since whenever \textit{Find} does not detect a copy of $H$ at least one vertex of $T$ is being deleted
(and as we said earlier, a copy of $H$ is found at most $M$ times).
Finally notice that the number of runs of \textit{Find} when two sets are being
output is constant (since the sets that are found by \textit{Find} are of linear size in $n$ and every time they are found they are deleted from $T_{1}$). The procedure of
\cite{mohri} is run only when \textit{Find} outputs two sets thus, according to what we have just said, this procedure
is run constant number of times. This observation and the remark that the success of the
procedure is independent of the input it acts on completes the proof.
\end{proof}

Theorem~4.1 follows now immediately from Theorem~\ref{clusttechtheorem} if we notice that
under the assumptions from the statement of Theorem~4.1, we have: $\delta_{M}=2$.

\subsection{Proof of Theorem~4.2}
\label{app_7}

We are ready to prove Theorem~4.2.

\begin{proof}
We have already proved Theorem~4.1 and as we will see now, this is main ingredient of the
proof of Theorem~4.2. Notice first that it suffices to show that procedure $\textit{Purify}$
outputs the set of all nonoutliers. Indeed, assume this is the case. Out of $N$ coming queries at least
$\frac{M}{M+n}$ (on average) will have both vertices from the same domain. At most an $\epsilon$-fraction of the set of that queries (on average) will have its first vertex in the set that was not partitioned by the \textit{DagClustering} algorithm and this will be also true for the second vertex. Finally, by the similar analysis, at least $2\epsilon$-fraction of the queries
with both vertices in the same domain and both partitioned by the algorithm will have at least one of its vertex in the set of
outliers of this domain. If we now take those queries for which this is not the case then it suffices to notice that 
the queries that do not correspond to backward edges in the ordering obtained by the algorithm and do not correspond
to backward edges in the canonical ordering of domains are answered correctly. Since the QuickSort algorithm produces a 3-approximation of the feedback arc set problem, we are done.\\

All we need to do is to prove the correctness of the \textit{Purify} procedure. 
Fix a part $\mathcal{P}_{i}$ of the partitioning and let $v \in \mathcal{P}_{i}$ be a vertex that is not an outlier.  Let $\Delta^{v}$ be the number of directed
triangles in $T^{c} | \mathcal{P}_{i}$ that are touching it. Let $\Delta^{v}_{1}$ be those of these triangles that
have at least one vertex in the set of outliers. Since $\mathcal{P}_{i}$ is $(1-\epsilon)$-pure, we trivially get:
$|\Delta^{v}_{1}| \leq \epsilon|\mathcal{P}_{i}|^{2}$. Let $\Delta^{v}_{2}$ be those triangles from $\Delta^{v}$
that have all three vertices in the set of nonoutliers. But every such triangle needs to have an edge that is backward under canonical ordering of the nonoutliers. From this we get: $\Delta^{v}_{2} \leq B^{\mathcal{P}_{i}} +
B^{\mathcal{P}_{i},v} \cdot \mathcal{P}_{i}$, where $B^{\mathcal{P}_{i}}$ is the set of backward edges under
canonical orderings of nonoutliers and $B^{\mathcal{P}_{i},v}$ is the set of backward edges under canonical orderings of nonoutliers with one endpoint in $v$.
Thus we get: $\Delta^{v} \leq \Delta^{v}_{1}  +\Delta^{v}_{2} \leq \epsilon|\mathcal{P}_{i}|^{2} + B^{\mathcal{P}_{i}} +
B^{\mathcal{P}_{i},v} \cdot \mathcal{P}_{i}$.
Now let $v$ be from the set of outliers. Let $\mathcal{R}$ denote the set of nonoutliers and let
$\mathcal{R}_{1}$ be its subset of first $\frac{|\mathcal{R}|}{2}$ vertices under canonical ordering and let
$\mathcal{R}_{2}$ be its subset of last $\frac{|\mathcal{R}|}{2}$ vertices under canonical ordering.
Notice that $\mathcal{R}_{1},\mathcal{R}_{2} \geq \frac{1-\epsilon}{2}|\mathcal{P}_{i}|$.
It is also easy to see that the number of directed triangles touching $v$ is at least:
$\Delta^{v} \geq |\mathcal{R}_{1}||\mathcal{R}_{1}| - B^{\mathcal{P}_{i}}$.
Thus, if both: $|\mathcal{R}_{1}| \geq \frac{(1-\epsilon)}{4}|\mathcal{P}_{i}|p_{m}$ and
$|\mathcal{R}_{2}| \geq \frac{(1-\epsilon)}{4}|\mathcal{P}_{i}|p_{m}$, we have:
$\Delta^{v} \geq \frac{(1-\epsilon)^{2}}{16}|\mathcal{P}_{i}|^{2}p_{m}^{2} - B^{\mathcal{P}_{i}}$.
Now it suffices to use Chernoff's inequality and the union bound, as we have done so far many times, to see that for the choice of $\epsilon$ from the statement of the theorem with probability $1-o(1)$ we have both:  
\begin{itemize}
\item $\Delta^{v} \geq \frac{(1-\epsilon)^{2}}{32}|\mathcal{P}_{i}|^{2}p_{m}^{2}$ for every outlier $v$, and
\item $\Delta^{v} < \frac{(1-\epsilon)^{2}}{32}|\mathcal{P}_{i}|^{2}p_{m}^{2}$ for every nonoutlier
\end{itemize}
as long as $|\mathcal{P}_{i}|$ is large enough.
We leave details to the reader this time.
Since $|\mathcal{P}_{i}|$ is linear in $n$ and to approximate $\Delta^{v}$ good enough with probability
$1-o(1)$ it trivially suffices to select $\Theta(\log(n))$ random samples, we are done.
\end{proof}

\subsection{Time complexity of the algorithms}
\label{app_8}

Let us analize time complexity of the \textit{DagClustering} algorithm first.
One run of the algorithm presented in \cite{mohri} requires $O(n\log(n))$ time on average 
(and this running time is highly concentrated
around its mean) but theoretically to be sure with probability $1-o(1)$ that the ranking that is found is a $3$-approximation
we need to perform it more than once. It suffices to perform it $\log(n)$ times (in practice 
it is not necessary to run it more than few times and this is what we did in our experiments).
We then output the ordering that gives the smallest number of backward edges.  This check will require
$O(n^{2})$ time. The mulitple run of the routine from \cite{mohri} is what we call QuickSort subroutine
in the algorithmic section of the main body of the paper. We have already noticed that the subroutine QuickSort
is called constant number of times in the \textit{DagClustering} algorithm.
We have already observed that \textit{Find} outputs two sets: $X$ and $Y$ constant number of times.
Assume that event $\mathcal{J}$ from the proof of Theorem~\ref{clusttechtheorem} holds.
We have also observed that \textit{Find} detects a copy of $H$ at most $M$ times (see:
the proof of Theorem~\ref{clusttechtheorem} for the definition of $M$).
Now, notice that a straightforward implementation of \textit{Find} requires $O(n^{2})$ time.
So conditioned on $\mathcal{J}$ the running time is $O(n^{2}M)$ with high probability.
$M$ is usually much smaller than $n$ thus the running time is slightly superquadratic. In practice it is
even close to linear due to several small heurstics that were used to speedup the entire algorithm (see: discussion
in the experimental section).
To see that the running time of the \textit{HeteroRanking} algorithm is also slightly superquadratic it suffices
to observe that the straightforward implementation of the \textit{Purify} procedure takes $O(n^{2})$ time.

\subsection{\textit{HeteroRanking} versus state-of-the-art ranking methods}
\label{app_9}

In this short subsection we would like to explain a little bit more quantitatively why in the heteregenous setting the algorithms such as \textit{QuickSort} and other methods that aim to find an ordering with small number of backward edges cannot succeed alone and need to act as an input that was previously preprocessed by some digraph clustering algorithm.
Let us focus on the \textit{QuickSort} algorithm first since it is very easy to implement.  
Let us take the very simple yet difficult enough for the \textit{QuickSort} algorithm setting of two
domains: $\mathcal{D}_{1}$ and $\mathcal{D}_{2}$. Assume that both are of the same size $\frac{n}{2}$.
Assume that for every $u \in \mathcal{D}_{1}$, $v \in \mathcal{D}_{2}$ there exists an edge $(u,v)$ in the preference
tournament $T$ with probability at least $p_{m}$ and there exists an edge $(v,u)$ in the preference
tournament $T$ with probability at least $p_{m}$. The \textit{QuickSort} chooses uniformly at random a 
pivot point $p$. By symmetry, assume without loss of generality that $p \in \mathcal{D}_{2}$. 
Let $\mathcal{D}_{1}^{1}$ be the set of first $\frac{|\mathcal{D}_{1}}{2}|$ vertices of $\mathcal{D}_{1}$ under its
canonical ordering and let $\mathcal{D}_{1}^{2}$ be the set of last $\frac{|\mathcal{D}_{1}}{2}|$ vertices of $\mathcal{D}_{1}$ under its canonical ordering. Let $N_{1}$ be the set of outneighbors of $p$ in $\mathcal{D}_{1}^{1}$
and let $N_{2}$ be the set of inneighbors of $p$ in $\mathcal{D}_{1}^{2}$.
Notice that under the first reordering of vertices in the \textit{QuickSort} algorithm all the points from $N_{2}$
will be ordered before all the points from $N_{1}$. Since every point from $N_{1}$ is adjacent to every point
from $N_{2}$, after first reordering of vertices we will produce at least $|N_{1}||N_{2}|$ backward edges.
Then obviously the number of backward edges will be at least $|N_{1}||N_{2}|$ (in fact it is easy to prove that
it will increase even more but lets take the simple bound we obtained from the first iteration).
Notice that the expected size of $N_{1}$ is $\frac{n}{2}p_{m}$ and this is also true for $N_{2}$.
Thus the average number of backward edges in the ranking output by the \textit{QuickSort} algorithm is at least
$\frac{n^{2}}{4}p_{m}$. What is even more important, the backward edges we were talking about so far had the property
that both their endpoints were taken from the same domain. Therefore it is easy to see that the obtained ranking is of very bad quality and will incorrectly answer a significant fraction of all coming queries. In particular, there is no chance
to obtain recall close to one. However, as we have already showed, a purely combinatorial digraph clustering mechanism combined with the \textit{Quicksort} algorithm enables to achieve it. The problem we raised above is not related only to the \textit{QuickSort} method. One can easily prove that for a tournament with $k$ domains of size $n$ each, where the directions of edges between different domains are chosen independently at random, the number of backward edges under every ordering is quadratic with probability close to $1$.  
Presented digraph clustering mechanism is crucial for filtering out low-quality information and detecting regions of much lower entropy that correspond to much denser regions of the underlying majority-voting model probability distribution $\mathcal{P}_{\mathcal{U}}$ . 

\subsection{Weighted setting}
\label{app_10}

Notice that in this setting we consider not just digraphs but tournaments. Besides we assume that preference tournaments are unweighted.  This however does not narrow
the generality of our analysis at all. All the results we obtained transform naturally to the weighted digraph setting.  However considering random model of the preference tournament
in fact enables us to accurately mimic digraph weighted setting, even without making any transformation. The lack of some edges in the general digraph setting was introduced
to emulate the scenario, where there are no statistics regarding some pair of objects or those statistics are very poor. This is straightforwardly simulated in our model by
edges between points from different domains. Both possible directions of those edges have significant probabilities of being chosen. In particular, if both are equal to
$\frac{1}{2}$ then the expected "signed weight" of the corresponding pair of points is $0$ which means than an edge is absent. In the general digraph model the weights
were introduced to emulate the fact that some statistics are more important or the users are more confident about preferences between some objects than others. All the
weights were takne from the interval [0,1]. But of course weights from that interval became probabilities in our model. Thus we do not lose anything by considering unweighted preference tournaments.

\subsection{Experiments}
\label{app_11}

We conducted several experiments to test the ranking mechanism of the \textit{HeteroRanking} 
algorithm as well
as the quality of the clustering produced by the \textit{DagClustering} procedure. 
We also compared our results with those
obtained by the state-of-the-art techniques. 

\begin{table*}[htp]
\centering
\caption{Table comparing the best ranking of the four constructed by: \cite{mohri}, \cite{balcan}, 
\cite{airola} and \cite{klautau} with the \textit{HeteroRanking} algorithm.
Tests were conducted for $C=15$, $depth=12$ and $V=100$ votes for every pair of objects within a groundtruth cluster.
Number of objects is given in $10^{3}$ units and \textit{ratio} in $10^{-2}$ units.} 
\label{table1}
\begin{tabular}{|c|c|c|c|c|c|c|c|c|c|c|c|c|c|c|c|c|c}
\hline
$n$ [in $10^{3}$] & 1  & 2 & 2.5 & 3 & 3.5 & 4 & 4.5 & 5 & 5.5 & 6 & 6.5     \\ \hline
$k$  & 2  & 2 & 2 & 2 & 3 & 3 & 3 & 3  & 4 & 4 & 4   \\ \hline
$ratio$ [in $10^{-2}$]  & 2  & 4 & 6 & 8 & 12 & 14 & 16 & 18  & 19 & 20 & 22   \\ \hline
$p_{succ}$  & 0.55  & 0.55 & 0.55 & 0.55 & 0.55 & 0.55 & 0.6 & 0.6  & 0.6 & 0.6 & 0.6    \\ \hline
$\epsilon_{best of four}$  & 0.33  & 0.27  & 0.28  & 0.30 & 0.24 & 0.33 & 0.35 & 0.32  & 0.34 & 0.35 & 0.34   \\ \hline
$\epsilon_{clust}$  & 0.09  & 0.12  & 0.16  & 0.14  & 0.18 & 0.17 & 0.13 & 0.15  & 0.17 & 0.2 & 0.22     \\ \hline
\end{tabular}
\end{table*}

Table 1 compares the quality of the ranking produced by the $\textit{HeteroRanking}$ algorithm with the best one from the following four: \cite{mohri}, \cite{balcan}, 
\cite{airola} and\\ \cite{klautau} .  
The results cover: different number of domains and quality characteristics of the statistics published according to the
majority-voting mechanism. We use the following notation: $n$ - number of all the objects,
$k$-number of groundtruth clusters, $ratio$ - the ratio between the number of votes for pairs of objects from different clusters
and from the same cluster, $p_{succ}$ - the probability that a voter will correctly classify a given pair of objects, 
$\epsilon_{best of four}$ - generalization error of the best of four state-of-the-art approaches, $\epsilon_{clust}$ - generalization error
of the \textit{HeteroRanking} algorithm. 

\begin{table*}[htp]
\centering
\caption{Table comparing the best ranking of the four constructed by: \cite{mohri}, \cite{balcan}, 
\cite{airola} and \cite{klautau} with the \textit{HeteroRanking} algorithm.
This time we also change the \textit{depth} parameter. Tests were conducted for $C=15$, $V=100$ and $n=6500$.}
\label{table2}
\begin{tabular}{|c|c|c|c|c|c|c|c|c|c|c|c|c|c|c|c|c|c}
\hline
$k$  & 2  & 2 & 2 & 2  & 3 & 3 & 3 & 3 & 4 & 4 & 4   \\ \hline
$ratio$ [in $10^{-2}$]  & 2  & 4 & 8 & 10  & 12 & 14 & 16 & 20  & 22 & 23 & 25  \\ \hline
$p_{succ}$  & 0.55  & 0.55 & 0.55 & 0.55  & 0.55 & 0.55 & 0.55 & 0.6  & 0.6 & 0.6 & 0.6   \\ \hline
$depth$  & 8  & 9  & 10  & 11  & 12 & 20 & 25 & 30 & 40 & 50 & 55    \\ \hline
$\epsilon_{best of four}$  & 0.35  & 0.33  & 0.35  & 0.32   & 0.31 & 0.34 & 0.34 & 0.28  & 0.24 & 0.23 & 0.23   \\ \hline
$\epsilon_{clust}$  & 0.35  & 0.36  & 0.30  & 0.20 & 0.14 & 0.12 & 0.12 & 0.16 & 0.18 & 0.20 & 0.22   \\ \hline
\end{tabular}
\end{table*}

\vspace{-0.05in}
\begin{table*}[b!]
\centering
\caption{Table comparing the best ranking of the four constructed by: \cite{mohri}, \cite{balcan}, 
\cite{airola} and \cite{klautau} with the \textit{HeteroRanking} algorithm.
Tests were conducted for $C=15$, $k=4$ and $V=100$. This time vertices are not uniformly splitted across the clusters and the sizes of the clusters are given as parameters:
$n_{1}$, $n_{2}$, $n_{3}$ and $n_{4}$.}
\label{table3}
\begin{tabular}{|c|c|c|c|c|c|c|c|c|c|c|c|c|c|c|c|c|c}
\hline
$n_{1}$ [in $10^{3}$]  & 0.5  & 0.5 & 0.5 & 0.5  & 0.6 & 0.6 & 0.6 & 0.6  & 0.6 & 0.7 & 0.7   \\ \hline
$n_{2}$ [in $10^{3}$] & 1  & 1 & 1 & 1  & 1 & 1 & 2 & 2  & 2 & 2 & 2   \\ \hline
$n_{3}$ [in $10^{3}$] & 2  & 2 & 2 & 2  & 2 & 2 & 2.5 & 2.5  & 2.5 & 3 & 3   \\ \hline
$n_{4}$ [in $10^{3}$]  & 1.5  & 1.5 & 1.7 & 1.7  & 1.8 & 1.8 & 1.8 & 1.9  & 1.9 & 2 & 2   \\ \hline
$ratio$ [in $10^{-2}$]  & 2  & 4 & 8 & 10  & 12 & 14 & 16 & 20  & 22 & 23 & 25   \\ \hline
$p_{succ}$  & 0.55  & 0.55  & 0.55  & 0.55   & 0.55 & 0.55 & 0.55 & 0.6  & 0.6 & 0.6 & 0.6   \\ \hline
$depth$  & 8  & 9  & 10  & 11   & 12 & 20 & 25 & 30  & 40 & 50 & 55   \\ \hline
$\epsilon_{best of four}$  & 0.35  & 0.31  & 0.29  & 0.28 & 0.19 & 0.18 & 0.17 & 0.19 & 0.23 & 0.22 & 0.19   \\ \hline
$\epsilon_{clust}$  & 0.37  & 0.35  & 0.20  & 0.1 & 0.12 & 0.13 & 0.13 & 0.20 & 0.25 & 0.20 & 0.21 \\ \hline
\end{tabular}
\end{table*}

If not explicitly stated otherwise then the vertices are uniformly splitted between clusters.
It was tested experimentally that adding a simple heuristic
to the \textit{Searcher} subprocedure of \textit{Find} algorithm can significantly
improve the running time without affecting accuracy. 
In the \textit{Searcher} we first randomly permute the 
vertices of the forbidden pattern establishing a random order in which we will look for them. To estimate whether the set of out/inneighbors of the given vertex $h_{i}$ is large enough we perform simple
sampling. 
Then, if we have already found $C$ copies of $H$ ($C$ is a parameter), and in the current run of \textit{Find}
we have found more than $d$ vertices of the potential embedding (we will call $d$ the \textit{depth} parameter)
we rerun \textit{Find}. As a gadget $H$ we use a random tournament of $60$ vertices since it was experimentally verified that this order of the tournament is good enough to obtain high-quality ranking.  

Table 2 compares our results with  the same methods as Table 1 from the main body of the paper, but for different values of  parameter \textit{depth}.
Additional results showing comparison of our approach with existing methods are presented in Table 3. The setting is similar to this for Table 1 and 2 
but this time domains are of different sizes.


Figure 1 and Figure 2 show how the generalization error depends on the quality characteristic of the statistics obtained from the
majority-voting mechanism. Algorithm \textit{HeteroRanking} outperforms state-of-the-art methods for statistics of
lower quality (smaller \textit{ratio} values).

\begin{figure}[h!]
  \centering
\includegraphics[width=0.46\textwidth,natwidth=610,natheight=642]{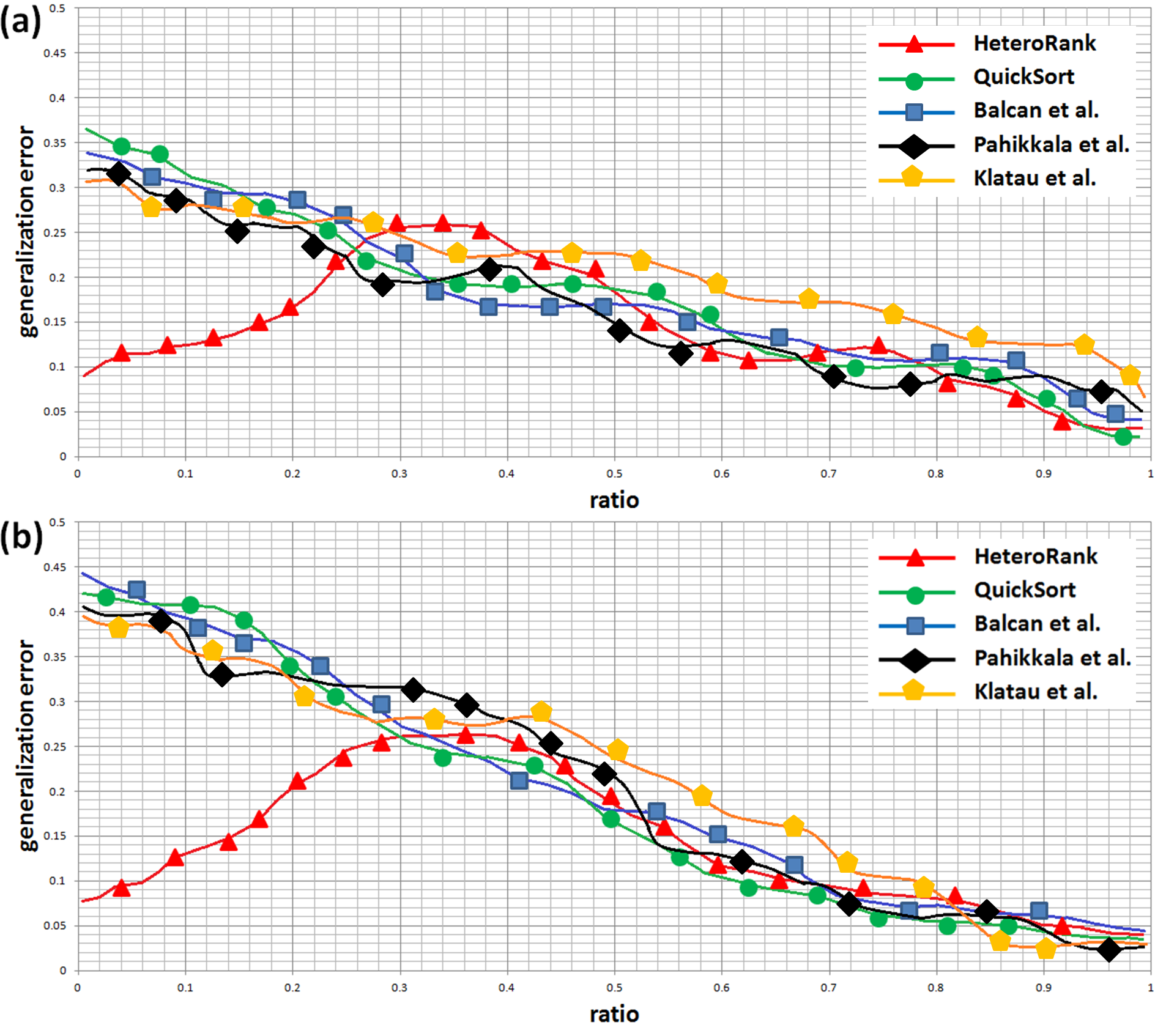}
  \caption{Diagrams comparing \textit{HeteroRanking} method with  
\cite{airola}, \cite{klautau}, \cite{balcan} and the \textit{QuickSort} algorithm from \cite{mohri}. Tests were performed for $C=15$, $n=7000$, $V=100$ and $depth=15$. The number of clusters is: (a) $k=3$, (b) $k=4$.}
\end{figure}

\begin{figure}[h!]
  \centering
\includegraphics[width=0.46\textwidth,natwidth=610,natheight=642]{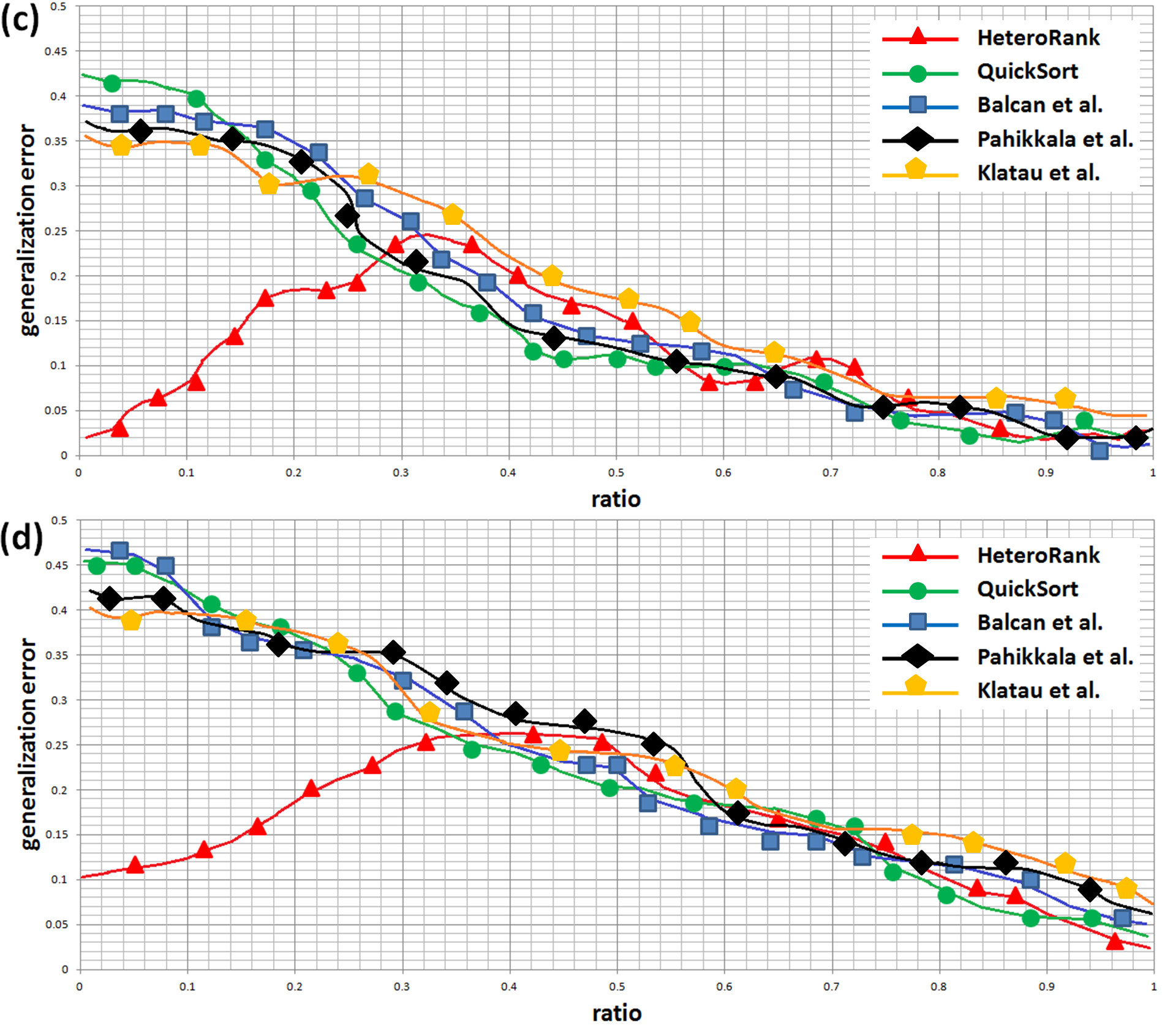}
  \caption{Diagrams comparing \textit{HeteroRanking} method with  
\cite{airola}, \cite{klautau}, \cite{balcan} and the \textit{QuickSort} algorithm from \cite{mohri}.  Tests were performed for $C=15$, $n=7000$, $V=100$ and $depth=15$. The number of clusters is: (c) $k=5$, (d) $k=6$.}
\end{figure}

We also performed experiments testing how many times in practice we need to run \textit{Find} procedure.
It turns out that the theoretical bounds we gave were very pesimistic and in fact the number of iterations is much smaller.
This implies much better running time. 
We checked experimentally that much smaller than assumed number of iterations
comes from the fact that in practice the sets $X,Y$ in the \textit{Find} procedure are detected much earlier and there
are also much larger (the results of the experiments are presented on Figure 3). 
Thus when the piece of the domain is being found in the \textit{HeteroRanking} algorithm, it is very large
on average. That in turn implies much faster reconstruction of the domain (up to the precision parameter $\epsilon$.)
We plan to investigate this phenomenon more closely from the theoretical point of view in the subsequent papers regarding the topic. 
We should also notice that, as was verified by us experimentally, the $\textit{Purify}$ subroutine does not necessarily 
need to be used to obtain good-quality ranking. Since the partitioning computed at earlier stages of the algorithm is very pure, the outliers are chosen as pivot points with very low probability and do not affect the overall quality of the ranking mechanism.

\begin{figure}[h!]
\label{fig:iterations}
  \centering
\includegraphics[width=0.6\textwidth,natwidth=610,natheight=642]{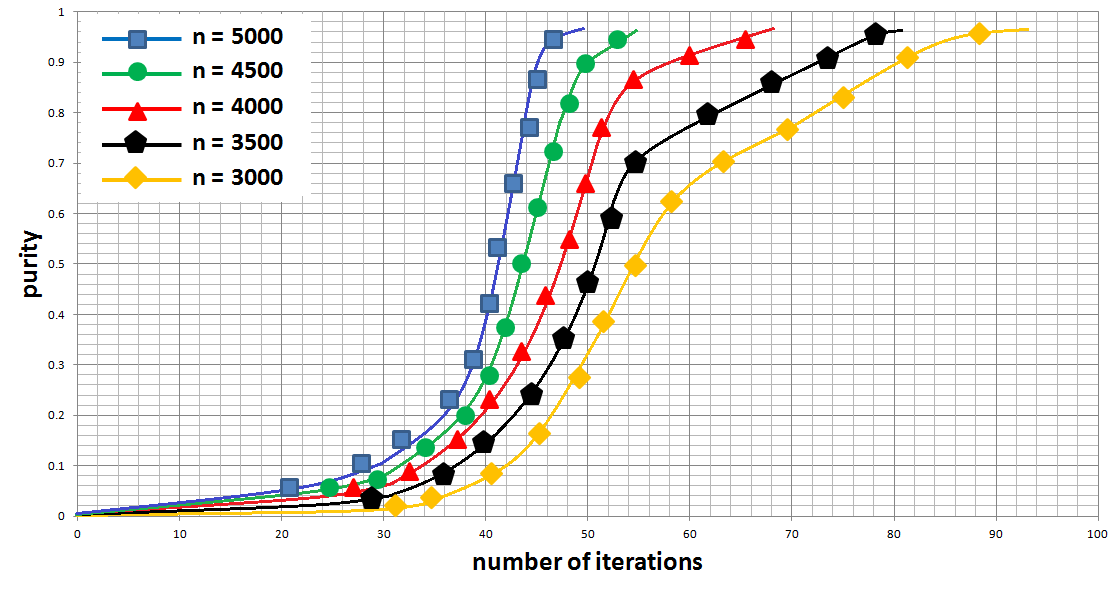}
  \caption{Diagram presenting how the purity of the computed clusters depends on the number of runs of the
\textit{Find} procedure which is the most expensive part of the \textit{HeteroRanking} algorithm.
The purity is defined as the fraction of the groundtruth domain that was already reconstructed. The tests were performed
for different sizes of the set of objects: $n=3000, 3500, 4000, 4500, 5000$, for $C=12$, $depth=14$,
$ratio=0.1$ and $V=200$.}
\end{figure}



\end{document}